\documentclass[aps, amsmath, amssymb, aps, pre,twocolumn,groupedaddress,superscriptaddress,nofootinbib,longbibliography]{revtex4-1}
\usepackage{balance}
\usepackage{amsfonts, amsmath, amssymb, xfrac}
\usepackage[colorlinks,linkcolor=black,citecolor=blue,urlcolor=black]{hyperref}
\usepackage{float}

\usepackage{amsthm}

\usepackage{enumitem}
\setlist{noitemsep,topsep=0pt,parsep=0pt,partopsep=0pt}
\usepackage{array}
\usepackage[english]{babel}

\usepackage{graphics}
\usepackage{graphicx}
\usepackage[config, labelfont={sf,bf}, textfont=sf]{caption,subfig}
\graphicspath{{./figs/}}
\usepackage[suffix=, outdir=./figs/]{epstopdf}
\usepackage{hhline}

\usepackage{lineno}

\modulolinenumbers[5]

\DeclareMathOperator{\LEs}{LE}
\DeclareMathOperator{\LCEs}{LCE}

\theoremstyle{definition}
\newtheorem{definition}{Definition}

\theoremstyle{remark}

\newtheorem{lemma}{Lemma}

\renewenvironment{proof}{\textit{Sketch of the proof.}}{\qed}

\newcommand{\mi}{\mathrm{i}}

\begin{document}

\title{Finite-time Lyapunov dimension and hidden attractor
 of the Rabinovich system
}

\author{N. V. Kuznetsov}
\email[]{Corresponding author email: nikolayv.kuznetsov@tdt.edu.vn, nkuznetsov239@gmail.com}
\affiliation{
Modeling Evolutionary Algorithms Simulation and Artificial Intelligence,
Faculty of Electrical \& Electronics Engineering,
Ton Duc Thang University, Ho Chi Minh, Vietnam}
\affiliation{Faculty of Mathematics and Mechanics, St. Petersburg State University,
Peterhof, St. Petersburg, Russia}
\affiliation{Department of Mathematical Information Technology,
University of Jyv\"{a}skyl\"{a}, Jyv\"{a}skyl\"{a}, Finland}
\author{G. A. Leonov}
\affiliation{Faculty of Mathematics and Mechanics, St. Petersburg State University,
Peterhof, St. Petersburg, Russia}
\affiliation{Institute of Problems of Mechanical Engineering RAS, Russia}
\author{T. N. Mokaev}
\affiliation{Faculty of Mathematics and Mechanics, St. Petersburg State University,
Peterhof, St. Petersburg, Russia}
\author{A. Prasad}
\affiliation{Department of Physics \& Astrophysics, Delhi University, India}
\author{M.D. Shrimali}
\affiliation{Central University of Rajasthan, Ajmer, India}

\date{\today}

\begin{abstract}
The Rabinovich system,
describing the process of interaction between waves in plasma, is considered.
It is shown that the Rabinovich system can exhibit a {hidden attractor}
in the case of multistability as well as a classical {self-excited attractor}.
The hidden attractor in this system can be localized
by analytical-numerical methods based on the {continuation} and {perpetual points}.
For numerical study of the attractors' dimension
the concept of {finite-time Lyapunov dimension} is developed.
A conjecture on the Lyapunov dimension of self-excited attractors
and the notion of {exact Lyapunov dimension} are discussed.
A comparative survey on the computation of the finite-time Lyapunov exponents
by different algorithms is presented
and an approach for a reliable numerical estimation of
the finite-time Lyapunov dimension is suggested.
Various estimates of the finite-time Lyapunov dimension
for the hidden attractor and hidden transient chaotic set
in the case of multistability are given.
\end{abstract}

\maketitle


\section{\label{sec:intro} Introduction}
One of the main tasks of the investigation of \emph{dynamical systems} is the study of established (limiting) behavior of the system after transient processes, i.e., the problem of localization and analysis of \emph{attractors} (limited sets of system's states,
which are reached by the system from close initial data after transient processes) \cite{Poincare-1892,Lyapunov-1892,LeonovR-1987}.
While trivial attractors (stable equilibrium points) can be easily found analytically,
the search of periodic and chaotic attractors can turn out to be a challenging problem
(see, e.g. famous 16th Hilbert problem \cite{Hilbert-1901}
on the number of coexisting periodic attractors in two dimensional polynomial systems,
which was formulated in 1900 and is still unsolved;
see also its generalization for multidimensional systems
with chaotic attractors \cite{LeonovK-2015-AMC}).
For numerical localization of an attractor 
one needs to choose an initial point in the basin of attraction and observe how the trajectory, starting from this initial point, after a transient process visualizes the attractor.
\emph{Self-excited attractors}, even coexisting in the case of \emph{multistability} \cite{PisarchikF-2014},
can be revealed numerically by the integration of trajectories,
started in small neighborhoods of unstable equilibria,
while \emph{hidden attractors} have the basins of attraction,
which are not connected with equilibria, and are hidden somewhere in the phase space
\cite{LeonovK-2013-IJBC,KuznetsovL-2014-IFACWC,LeonovKM-2015-EPJST,Kuznetsov-2016}.
Remark that in numerical computation of trajectory over a finite-time interval
it is difficult to distinguish a \emph{sustained chaos} from
a \emph{transient chaos}
(a transient chaotic set in the phase space,
which can nevertheless persist for a long time) \cite{GrebogiOY-1983}.
Thus, the search and visualization of hidden attractors
and transient sets in the phase space are challenging tasks
\cite{DudkowskiJKKLP-2016}.

In this paper we study hidden attractors and transient chaotic sets
in the Rabinovich system.
We show that the methods of numerical continuation and perpetual point
are helpful for localization and understanding of
hidden attractor in the Rabinovich system.

For the study of chaotic sets and dimension of attractors
the concept of the Lyapunov dimension \cite{KaplanY-1979}
was found useful and became widely spread
\cite{GrassbergerP-1983,EckmannR-1985,ConstantinFT-1985,AbarbanelBST-1993,BoichenkoLR-2005}.
Since in numerical experiments we can consider only finite time,
in this paper we develop the concept of the \emph{finite-time Lyapunov dimension} \cite{Kuznetsov-2016-PLA}
and an approach for its reliable numerical computation.
Various estimates of the finite-time Lyapunov dimension for the Rabinovich hidden attractor
in the case of multistability are given.

\section{\label{sec:intro} The Rabinovich system: interaction between waves in plasma}

Consider a system, suggested in 1978 by M.~Rabinovich~\cite{Rabinovich-1978,PikovskiRT-1978},
\begin{equation}
\begin{aligned}
	& \dot{x}  =  h y - \nu_1 x - y z, \\
	& \dot{y}  =  h x - \nu_2 y + x z, \\
	& \dot{z}  =  - z + x y,
\end{aligned}
\label{sys:rabinovich}
\end{equation}
describing the interaction of three resonantly coupled waves,
two of which are parametrically excited.
Here, the parameter $h$ is proportional to the
pumping amplitude and the parameters $\nu_{1,2}$
are normalized dumping decrements.

After the linear transformation (see, e.g., \cite{LeonovB-1992}):
\begin{equation}\label{xyzchange}
	\chi: (x,y,z) \to (\nu_1 \nu_2 h^{-1} y, \nu_1 x, \nu_1 \nu_2 h^{-1} z)
\end{equation}
and time rescaling:
\begin{equation}\label{timechange}
  	t \to \nu_1^{-1} t,
\end{equation}
we obtain a generalized Lorenz system:
\begin{equation}
\begin{aligned}
 & \dot{x}  =  - \sigma(x - y) - a y z,\\
 & \dot{y}  =  r x - y - x z,\\
 & \dot{z}  =  -b z + x y,
\end{aligned}
\label{sys:lorenz-general}
\end{equation}
where
\begin{equation}
	\sigma = \nu_1^{-1} \nu_2,\,
	b = \nu_1^{-1},\,
	a = -\nu_2^2 h^{-2},\,
	r = \nu_1^{-1} \nu_2^{-1}h ^{2}.
	\label{eq:params-relation}
\end{equation}

System \eqref{sys:lorenz-general} with  $a = 0$
coincides with the classical Lorenz system \cite{Lorenz-1963}.
As it is discussed in \cite{LeonovB-1992},
system \eqref{sys:lorenz-general}
can also be used to describe the following physical processes:
the convective fluid motion inside rotating ellipsoid, 
the rotation of rigid body in viscous fluid, 
the gyrostat dynamics, 
the convection of horizontal layer of fluid making harmonic oscillations, 
and the model of Kolmogorov's flow. 

Note that since parameters $\nu_1$, $\nu_2$, $h$ are positive,
the parameters $\sigma$, $b$, $r$ are positive and parameter $a$ is negative.
From relation \eqref{eq:params-relation} we have:
\begin{equation}\label{cond:rabinovich}
  \sigma = -ar.
\end{equation}

Further, we study system \eqref{sys:lorenz-general}
under the assumption \eqref{cond:rabinovich}.
If $r < 1$, then system \eqref{sys:lorenz-general}
has a unique equilibrium ${\bf \rm S_0} = (0,0,0)$, which is
globally asymptotically Lyapunov stable (global attractor)
\cite{LeonovB-1992,BoichenkoLR-2005}.
If $r > 1$, then system \eqref{sys:lorenz-general}
has three equilibria: ${\bf \rm S_0} = (0,0,0)$ and
$
  {\bf \rm S_{\pm}} = (\pm x_1, \, \pm y_1, \, z_1),
$
where
\[
 x_1 = \frac{\sigma b \sqrt{\xi}}{\sigma b + a \xi}, \quad
 y_1 = \sqrt{\xi}, \quad
 z_1 = \frac{\sigma \xi}{\sigma b + a \xi},
\]
and
\[
 \xi = \frac{\sigma b}{2 a^2} \left[ a (r-2) - \sigma + \sqrt{(\sigma - ar)^2 + 4a\sigma} \right].
\]
The stability of equilibria $S_{\pm}$ of system \eqref{sys:lorenz-general} depends on
the parameters $r$, $a$, and $b$.
Using the Routh-Hurwitz criterion, we obtain the following
\begin{lemma}\label{prop:stability}
  The equilibria $S_{\pm}$ of system \eqref{sys:lorenz-general} with
  parameters \eqref{eq:params-relation} are stable
  if and only if one of the following conditions holds:
  \begin{enumerate}[label=(\roman*)]
    \item $0 \leq \, ar + 1 \, < \frac{2 r}{r - \sqrt{r(r-1)}}$, \label{cond:stability:1}
    \item $ar + 1 < 0$, \, $b > b_{\rm cr} =
    \frac{4 a (r - 1) (ar + 1) \sqrt{r(r-1)} + (ar - 1)^3}{(ar + 1)^2 - 4ar^2}.$
    \label{cond:stability:2}
  \end{enumerate}
\end{lemma}
\begin{proof}
  The coefficients of the characteristic polynomial
  $\chi(x,\,y,\,z) = \lambda^3 +  p_1(x,\,y,\,z) \lambda^2 + p_2(x,\,y,\,z) \lambda + p_3(x,\,y,\,z)$
  of the Jacobian matrix of system \eqref{sys:lorenz-general}
  at the point $(x,\,y,\,z)$ are the following
  \begin{align*}
    & p_1(x,\,y,\,z) =  b - ar + 1, \\
    & p_2(x,\,y,\,z) = x^2 + a y^2 - a z^2 - ar (b - r + 1)+ b, \\
    & p_3(x,\,y,\,z) = - a \big(2 x y z + r x^2 - y^2 + b z^2  - b r (r-1)\big).
  \end{align*}
  One can check that inequalities $p_1(x_1,\,y_1,\,z_1) > 0$ and $p_3(x_1,\,y_1,\,z_1) > 0$ are always valid.
  If $ar + 1 \geq 0$, then
  $p_1(x_1,\,y_1,\,z_1) p_2(x_1,\,y_1,\,z_1) - p_3(x_1,\,y_1,\,z_1) > 0$,
  and if condition \ref{cond:stability:1} also holds, then $p_2(x_1,\,y_1,\,z_1) > 0$.

  If $ar + 1 < 0$, then $p_2(x_1,\,y_1,\,z_1) > 0$, and
  if condition \ref{cond:stability:2} also holds, then
  $p_1(x_1,\,y_1,\,z_1) p_2(x_1,\,y_1,\,z_1) - p_3(x_1,\,y_1,\,z_1) > 0$.
\end{proof}

\section{\label{sec:attractor} Attractors and transient chaos}

Consider system \eqref{sys:lorenz-general}
as an autonomous differential equation of a general form:
\begin{equation}\label{sys:ode}
  \dot{u} = f({u}),
\end{equation}
where $u=(x,y,z) \in \mathbb{R}^3$, and the
continuously differentiable vector-function $f: \mathbb{R}^3 \to \mathbb{R}^3$
represents the right-hand side of system \eqref{sys:lorenz-general}.
Define by ${u}(t,{u}_0)$ a solution of \eqref{sys:ode} such that
${u}(0,{u}_0)={u}_0$.
For system \eqref{sys:ode}, a bounded closed invariant set $K$ is
\begin{enumerate}[label=(\roman*)]
  \item a {\it (local)  attractor} if it is a minimal locally attractive set
        (i.e. $\lim_{t \to +\infty} {\rm dist} (K, {u}(t,{u}_0)) = 0$
        $\forall {u_0} \in K(\varepsilon)$, where $K(\varepsilon)$ is a
        certain $\varepsilon$-neighborhood of set $K$),
  \item a {\it global attractor} if it is a minimal globally attractive set
        (i.e. $\lim_{t \to +\infty} {\rm dist} (K, {u}(t,{u}_0)) = 0$
        $\forall {u_0} \in \mathbb{R}^3$),
\end{enumerate}
where ${\rm dist}(K, {u}) = \inf_{{v} \in K} ||{v} - {u}||$
is the distance from the point ${u} \in \mathbb{R}^3$ to the set $K \subset \mathbb{R}^3$ (see, e.g. \cite{LeonovKM-2015-EPJST}).

Note that system \eqref{sys:lorenz-general} (or \eqref{sys:ode})
is dissipative in the sense that it possesses a bounded convex absorbing set
\cite{LeonovB-1992,LeonovKM-2015-EPJST}:
\begin{equation}\label{absorb_set}
  \mathcal{B}(r,a,b) = \left\{u \in \mathbb{R}^3 ~|~ V(u) \leq \frac{b (\sigma + \delta r)^2}{2 c (a + \delta)} \right\},
\end{equation}
where $V(u)=V(x,y,z) = x^2 + \delta y^2 + (a + \delta)\left(z - \frac{\sigma + \delta r}{a + \delta}\right)^2$,
$\delta$ is an arbitrary positive number such that $a + \delta > 0$ and $c = \min(\sigma, 1, \frac{b}{2})$.
Thus, the solutions of \eqref{sys:lorenz-general} exist for $t \in [0,+\infty)$
and system \eqref{sys:lorenz-general}
possesses a global attractor \cite{Chueshov-2002-book,LeonovKM-2015-EPJST},
which contains the set of all equilibria
and can be constructed as
$\cap_{\tau > 0} \overline{\cup_{t \geq \tau} \varphi^t\left(\mathcal{B}\right)}$.

Computational errors (caused by a finite precision arithmetic and numerical integration
of differential equations) and sensitivity to initial data 
allow one to get a reliable visualization of a \emph{chaotic attractor}
by only one pseudo-trajectory computed for a sufficiently large time interval.
One needs to choose an initial point in the basin of attraction of the attractor
and observe how the trajectory, starting from this initial point,
after a transient process visualizes the attractor.
Thus, from a computational point of view, it is natural
to suggest the following classification of attractors,
based on the simplicity of finding the basins of attraction in the phase space.

\begin{definition}{\cite{LeonovKV-2011-PLA,LeonovK-2013-IJBC,LeonovKM-2015-EPJST,Kuznetsov-2016}}
 An attractor is called a \emph{self-excited attractor}
 if its basin of attraction
 intersects with any open neighborhood of an equilibrium,
 otherwise, it is called a \emph{hidden attractor}.
\end{definition}

For a \emph{self-excited\footnote{
The term \emph{self oscillation}
(selbsterregten Schwingungen in German)
can be traced back to the works of Barkhausen and Andronov,
where it was used to describe the generation and maintenance of a periodic motion
in electromechanical models by a source of power that lacks any corresponding periodicity
(e.g., a stable limit cycle in the van der Pol oscillator) \cite{Barkhausen-1935,MandelstamP-1932,AndronovVKh-1937,Jenkins-2013}.
}
attractor} its basin of attraction
is connected with an unstable equilibrium
and, therefore, self-excited attractors
can be localized numerically by the
\emph{standard computational procedure}
in which after a transient process a trajectory,
starting in a neighborhood of an unstable equilibrium,
is attracted to the state of oscillation and then traces it.
Thus, self-excited attractors can be easily visualized
(e.g. the classical Lorenz, R\"{o}ssler, and H\'{e}non  attractors
are self-excited with respect to unstable zero equilibrium
and can be easily visualized by a trajectory from its vicinity).

For a hidden attractor, its basin of attraction is not connected with equilibria
and, thus, the search and visualization of hidden attractors in the phase space may be a challenging task.
Hidden attractors are attractors in the systems without equilibria
(see, e.g. rotating electromechanical systems with Sommerfeld effect (1902)
\cite{Sommerfeld-1902,KiselevaKL-2016-IFAC}),
and in the systems with only one stable equilibrium
(see, e.g. counterexamples \cite{LeonovK-2011-DAN,LeonovK-2013-IJBC}
to Aizerman's (1949) and Kalman's (1957) conjectures
on the monostability of  nonlinear control systems
\cite{Aizerman-1949,Kalman-1957}).
One of the first related problems is the second part
of 16th Hilbert problem \cite{Hilbert-1901}
on the number and mutual disposition of limit cycles
in two dimensional polynomial systems,
where nested limit cycles (a special case of multistability and coexistence of periodic attractors)
exhibit hidden periodic attractors (see, e.g., \cite{Bautin-1939,KuznetsovKL-2013-DEDS,LeonovK-2013-IJBC}).
The \emph{classification of attractors as being hidden or self-excited}
was introduced by Leonov \& Kuznetsov
in connection with the discovery of the first hidden Chua attractor 
\cite{KuznetsovLV-2010-IFAC,LeonovKV-2011-PLA,BraginVKL-2011,LeonovKV-2012-PhysD,KuznetsovKLV-2013,KiselevaKKKLYY-2017,StankevichKLC-2017}
and has captured much attention of scientists from around the world
(see, e.g. \cite{BurkinK-2014-HA,LiSprott-2014-HA,LiZY-2014-HA,PhamRFF-2014-HA,
ChenLYBXW-2015-HA,KuznetsovKMS-2015-HA,SahaSRC-2015-HA,SemenovKASVA-2015,SharmaSPKL-2015-EPJST,ZhusubaliyevMCM-2015-HA,
DancaKC-2016,JafariPGMK-2016-HA,MenacerLC-2016-HA,OjoniyiA-2016-HA,PhamVJVK-2016-HA,RochaM-2016-HA,WeiPKW-2016-HA,Zelinka-2016-HA,
BorahR-2017-HA,BrzeskiWKKP-2017,FengP-2017-HA,JiangLWZ-2016-HA,KuznetsovLYY-2017-CNSNS,MaWJZH-2017,MessiasR-2017-HA,SinghR-2017-HA,VolosPZMV-2017-HA,WeiMSAZ-2017-HA,ZhangWWM-2017-HA}).


Since in the numerical computation of trajectory over a finite-time interval
it is difficult to distinguish a \emph{sustained chaos} from
a \emph{transient chaos}
(a transient chaotic set in the phase space, which can nevertheless persist for a long time)
\cite{GrebogiOY-1983,LaiT-2011},
a similar to the above classification can be introduced for
the transient chaotic sets.

\begin{definition}{\cite{DancaK-2017-CSF,ChenKLM-2017-IJBC}}
 A \emph{transient chaotic set} is called a \emph{hidden transient chaotic set}
 if it does not involve and attract trajectories
 from a small neighborhood of equilibria;
 otherwise, it is called \emph{self-excited}.
\end{definition}

In order to distinguish an attracting chaotic set (attractor)
from a transient chaotic set in numerical experiments,
one can consider a grid of points in a small neighborhood of the set
and check the attraction of corresponding trajectories towards the set.


For system \eqref{sys:rabinovich} with parameters
$\nu_1 = 1$, $\nu_2 = 4$, and increasing $h$
it is possible to observe~\cite{PikovskiRT-1978} the classical scenario of
transition to chaos (via homoclinic and subcritical Andronov-Hopf bifurcations)
similar to the scenario in the Lorenz system.
For $4.84 \lessapprox h \lessapprox 13.4$
in system \eqref{sys:rabinovich} there is a self-excited chaotic attractor
(see e.g. Fig.~\ref{fig:rabinovich:attr:SE}), which coexists with two stable equilibria.
The same scenario can be obtained for system \eqref{sys:lorenz-general}
when parameters $b > 0$ and $a < 0$ are fixed and $r$ is increasing.
Besides self-excited chaotic attractors,
a hidden attractor was found in the system
\cite{KuznetsovLMS-2016-INCAAM,ChenKLM-2017-IJBC}.
Note that in \cite{LeonovKM-2015-CNSNS,LeonovKM-2015-EPJST}
system \eqref{sys:lorenz-general} with $a > 0$ was studied and a hidden attractor
was also found numerically.

In this work we localize a hidden chaotic attractor
in system \eqref{sys:lorenz-general} with $a < 0$
by the numerical continuation method starting from a self-excited chaotic attractor.
We change parameters, considered in \cite{KuznetsovLMS-2016-INCAAM},
in such a way that the chaotic set locates
not too close to the unstable zero equilibrium
to avoid a situation, when numerically integrated trajectory
oscillates for a long time and then
falls on the unstable manifold of unstable zero equilibrium,
leaves the chaotic set, and tends to one of the stable equilibria.

\begin{figure}[h]
 \centering
 \includegraphics[width=0.45\textwidth]{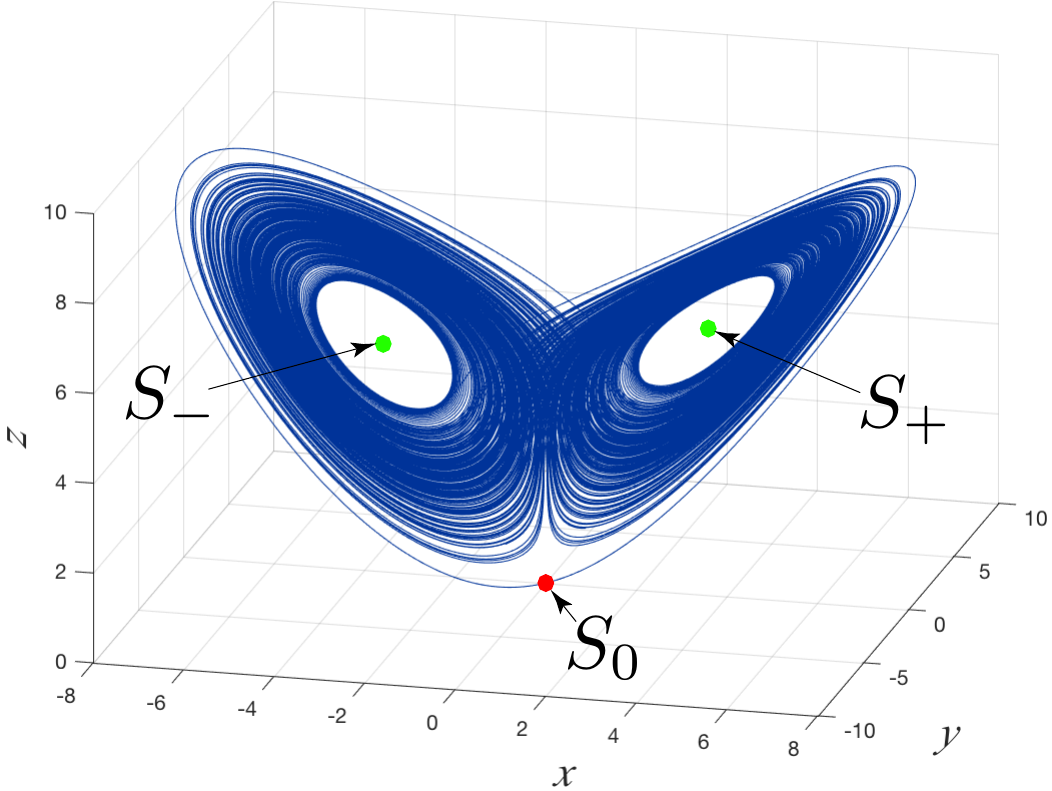}
 \caption{\label{fig:rabinovich:attr:SE}
 Multistability in the Rabinovich system \eqref{sys:lorenz-general}
 with 
 the classical values of parameters
 $\nu_1 = 1$, $\nu_2 = 4$, $h = 4.92$ from \cite{PikovskiRT-1978}:
 coexistence of three local attractors ---
 two stable equilibria $S_{\pm}$  and a chaotic self-excited attractor
 (self-excited with respect to the unstable zero equilibrium $S_0$).}
\end{figure}

\subsection{Localization via numerical continuation method} 

One of the effective methods for numerical localization of hidden attractors
in multidimensional dynamical systems is based on the
{\it homotopy} and {\it numerical continuation method (NCM)}.
The idea is to construct a sequence of
similar systems such that for the first (starting) system
the initial point for numerical computation of oscillating
solution (starting attractor)
can be obtained analytically, e.g,
it is often possible to consider the starting
system with a self-excited starting attractor;
then the transformation of this starting attractor
in the phase space
is tracked numerically while passing from one system to another;
the last system corresponds to the system
in which a hidden attractor is searched.

For the study of the scenario of transition to chaos,
we consider system \eqref{sys:ode} with $f(u) = f(u, \lambda)$,
where $\lambda \in \Lambda \subset \mathbb{R}^d$
is a vector of parameters whose variation in the parameter space $\Lambda$
determines the scenario.
Let $\lambda_{\rm end} \in \Lambda$ define a point corresponding to the system,
where a hidden attractor is searched.
Choose a point $\lambda_{\rm begin} \in \Lambda$ such that
we can analytically or numerically localize a certain nontrivial (oscillating) attractor
$\mathcal{A}^1$ in system \eqref{sys:ode} with $\lambda = \lambda_{\rm begin}$
(e.g., one can consider an initial self-excited attractor, defined by a trajectory ${u}^1(t)$
numerically integrated on a sufficiently large time interval $t \in [0, T]$
with the initial point ${u}^1(0)$ in the vicinity of an unstable equilibrium).
Consider a {\it path}\footnote{
  In the simplest case, when $d = 1$, the path is a line segment.
} in the parameter space $\Lambda$ , i.e. a continuous function
$\gamma~:~ [0,\,1] \to \Lambda$, for which $\gamma(0) = \lambda_{\rm begin}$ and
$\gamma(1) = \lambda_{\rm end}$,
and a sequence of points $\{\lambda^j\}_{j=1}^k$ on the path,
where $\lambda^1 = \lambda_{\rm begin}$,
$\lambda^k = \lambda_{\rm end}$,
such that the distance between
$\lambda^j$ and $\lambda^{j+1}$
is sufficiently small.
On each next step of the procedure,
the initial point for a trajectory to be integrated
is chosen as the last point of the trajectory integrated on the previous step:
${u}^{j+1}(0) = {u}^{j}(T)$.
Following this procedure and sequentially increasing $j$,
two alternatives are possible:
the points of $\mathcal{A}^j$ are in the basin of attraction
of attractor $\mathcal{A}^{j+1}$
or, while passing from system \eqref{sys:ode} with $\lambda = \lambda^j$
to system \eqref{sys:ode} with $\lambda = \lambda^{j+1}$,
a loss of stability bifurcation is observed and attractor $\mathcal{A}^j$ vanishes.
If while changing $\lambda$ from $\lambda_{\rm begin}$ to $\lambda_{\rm end}$
there is no loss of stability bifurcation of the considered attractors,
then a hidden attractor for $\lambda^k = \lambda_{\rm end}$ (at the end of the procedure)
is localized.

\subsection{Localization using perpetual points}

The equilibrium points of a dynamical system are the ones,
where the velocity and acceleration of the system simultaneously become zero.
If the existing equilibrium points are unstable,
then we may get either oscillating
or unbounded solutions.
In this section we show numerical results, which suggest that there are points,
termed as \emph{perpetual points} \cite{Prasad-2015},
which may help to visualize hidden attractors.

For system \eqref{sys:ode}, the equilibrium points ${u}_{\rm ep}$ are defined by
the equation $\dot{u} = f({u}_{\rm ep}) = 0$.
Consider a derivative of system \eqref{sys:ode} with respect to time
\begin{equation}\label{eq:acc}
  \ddot{u} = J({u}) \, f({u}) = g({u}),
\end{equation}
where
$J(u) =
\left[
 \frac{\partial f_i(u)}{\partial u_j}
\right]_{i,j=1}^n$
is the $n \times n$ Jacobian matrix.
Here $g({u})$ may be termed as an \emph{acceleration vector}.
System \eqref{eq:acc} shows the variation of acceleration in the phase space.

Similar to the equilibrium points estimation, where we set the velocity vector to zero,
we can also get a set of points, where $\ddot{u} = g({u}_{\rm pp}) = 0$
in \eqref{eq:acc}, i.e. the points corresponding to the zero acceleration.
At these points the velocity $\dot{u}$ may be either zero or nonzero.
This set includes the equilibrium points ${u}_{\rm ep}$ with zero velocity
as well as a subset of points with nonzero velocity.
These nonzero velocity points ${u}_{\rm pp}$ are termed as {\it perpetual points}
\cite{Prasad-2015,DudkowskiPK-2015-HA,Prasad-2016,DudkowskiJKKLP-2016}.
The reason why perpetual points may lead to hidden states ({\it perpetual point method (PPM)})
is still not well understood 
(see, e.g. discussion in \cite{NazarimehrSJS-2017}).

\begin{lemma}\label{lemma:perpetual}
  Perpetual points $S_{\rm pp} = (x_{\rm pp}, y_{\rm pp}, z_{\rm pp})$
  of system \eqref{sys:lorenz-general} can be derived from the following
  system
  \begin{equation}
  \left\{
  	\begin{aligned}
  		& (x_{\rm pp}^2 - a z_{\rm pp}^2 + a r^2 - 1 )
  		\left(-y_{\rm pp}^2 + z_{\rm pp}^2 + r^2 (a - 1) \right) =  \\
  		& \qquad\qquad\qquad\quad = r^2 (a r - 1)^2 - z_{\rm pp}^2 (a r - b - 1)^2, \\
  		& 2 x_{\rm pp} y_{\rm pp} z_{\rm pp} + r x_{\rm pp}^2 - y_{\rm pp}^2
  		+ b z_{\rm pp}^2 - b r ( r - 1 ) = 0, \\
  		& z_{\rm pp} = \frac{r x_{\rm pp}^2
  		 + (a r - b - 1) x_{\rm pp} y_{\rm pp} - a r y_{\rm pp}^2}{x_{\rm pp}^2 + a y_{\rm pp}^2 - b^2}.
  	\end{aligned}
  \right.
  \end{equation}
\end{lemma}

\section{\label{sec:attractor:hidden} Hidden attractor in \\ the Rabinovich system}

Next we apply the NCM for
localization of a hidden attractor in the Rabinovich system \eqref{sys:lorenz-general} and
check whether the attractor can be also localized using PPM.

In this experiment, we fix parameter $r$ and, using condition \ref{cond:stability:2}
of Lemma~\ref{prop:stability}, define parameters $a = - \frac{1}{r} - \varepsilon_1$
and $b = b_{\rm cr} - \varepsilon_2$.
For $r = 100$, $\varepsilon_1 = 10^{-3}$, $\varepsilon_2 = 10^{-2}$
we obtain $a = a_0 \equiv -1.1 \cdot 10^{-2}$, $b = b_0 \equiv 6.7454 \cdot 10^{-2}$ and
take $P_0(a_0, \, b_0)$ as the initial point
of line segment on the plane $(a, \,b)$.
The eigenvalues of the Jacobian matrix at the
equilibria $S_0$, $S_\pm$ of system~\eqref{sys:lorenz-general}
for these parameters are the following:
\begin{align*}
S_0 \, & : \quad 9.4382, \quad -0.0675, \quad -11.5382, \\
S_{\pm} \,  & : \quad 0.0037 \pm 3.6756 \, \mi, \quad -2.1749.
\end{align*}

Consider on the plane $(a,\, b)$ a line segment,
intersecting a boundary of stability domain of the equilibria $S_{\pm}$,
with the final point $P_2(a_2, \, b_2)$,
where $a_2 = a_0 + 1.035~\cdot~10^{-3} = -9.965 \cdot 10^{-3}$,
$b_2 = b_0 + \varepsilon_2 = 7.7454 \cdot 10^{-2}$,
i.e. the equilibrium $S_0$ remains saddle and the equilibria $S_\pm$
become stable focus-nodes
\begin{align*}
S_0 \, & : \quad 8.9842, \quad -0.0775, \quad -10.9807, \\
S_{\pm} \,  & : \quad -0.0401 \pm 3.9152 \, \mi, \quad -1.9937.
\end{align*}
\begin{figure}[!h]
	\centering
	\includegraphics[width=0.45\textwidth]{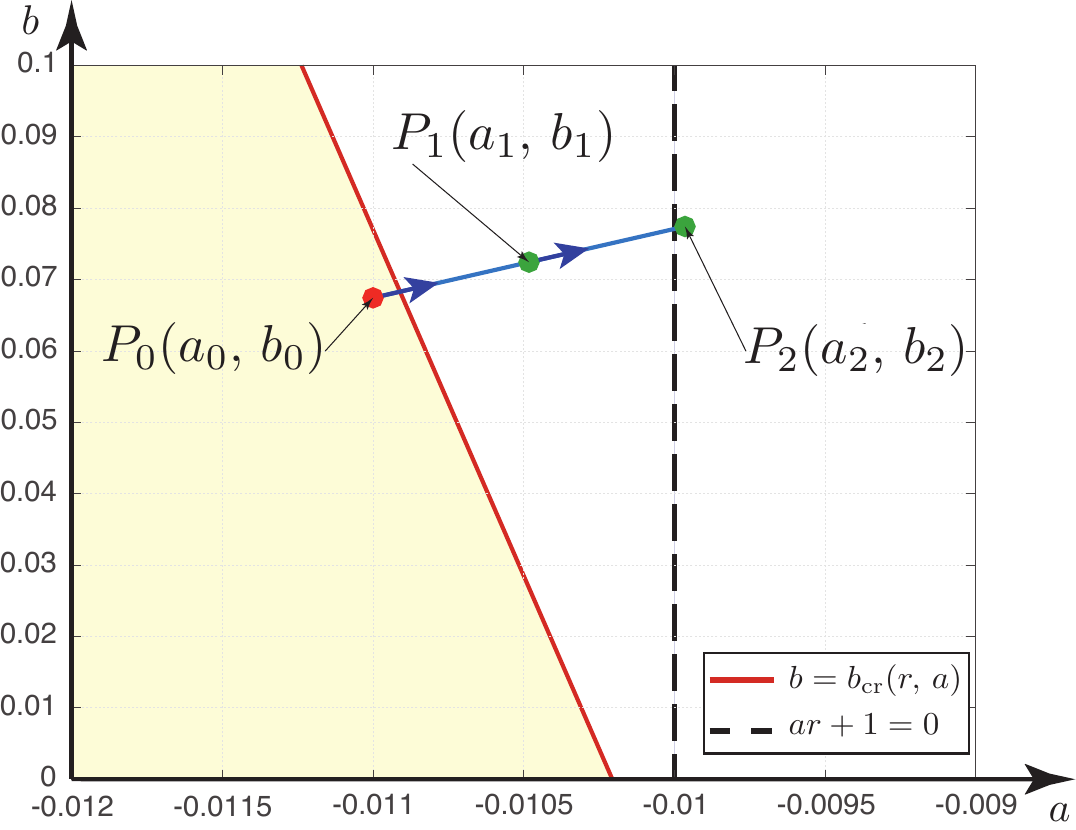}
	\caption{\label{fig:rabinovich:path}
  Path $P_0(a_0, \, b_0) \to P_1(a_1, \, b_1) \to P_2(a_2, \, b_2)$
  in parameters plane $(a,\,b)$ for the localization of hidden attractor in
  system \eqref{sys:lorenz-general} with $r = 100$.
  Here $a_0 = -1.1 \cdot 10^{-2}$, $b_0 = 6.7454 \cdot 10^{-2}$,
  $a_1 = -1.049 \cdot 10^{-2}$, $b_1 = 7.2454 \cdot 10^{-2}$,
  $a_2 = -9.965 \cdot 10^{-3}$, $b_2 = 7.7454 \cdot 10^{-3}$;
	$(\bullet)$ $P_0(a_0, \, b_0)$: self-excited attractor with respect to $S_{0}$, $S_{\pm}$;
  $(\bullet)$ $P_1(a_1, \, b_1)$: self-excited attractor with respect to $S_{0}$;
  $(\bullet)$ $P_2(a_2, \, b_2)$: hidden attractor.
  Stability domain is defined using Lemma~\ref{prop:stability}.
  }
\end{figure}

The initial point $P_0(a_0, \, b_0)$
corresponds to parameters
for which in system \eqref{sys:lorenz-general}
there exists a self-excited attractor.
Then for the considered line segment a sufficiently small partition step
is chosen and at each iteration step of the procedure
an attractor in the phase space of system \eqref{sys:lorenz-general}
is computed.
The last computed point at each step is used as the initial point
for the computation at the next step.
In this experiment we use NCM with $3$ steps on the path
$P_0(a_0, \, b_0) \to P_1 (a_1, \, b_1)\to P_2 (a_2, \, b_2)$, with
$a_1 = \frac{1}{2}(a_0 + a_2)$, $b_1 = \frac{1}{2}(b_0 + b_2)$
(see Fig.~\ref{fig:rabinovich:path}).
At the first step we have self-excited attractor
with respect to unstable equilibria $S_0$ and $S_\pm$;
at the second step the equilibria $S_\pm$
become stable but the attractor remains self-excited with respect to equilibrium $S_0$;
at the third step it is possible to visualize a hidden attractor of
system \eqref{sys:lorenz-general} (see Fig.~\ref{fig:rabinovich:attr:hidden}).

Using Lemma~\ref{lemma:perpetual} for parameters $r = 100$, $a = -9.965 \cdot 10^{-3}$,
$b = 7.7454 \cdot 10^{-2}$,
we obtain one perpetual point $S_{\rm pp} = (-0.2385, 49.1403, -101.4613)$,
which allows one to localize a hidden attractor (see Fig.~\ref{fig:rabinovich:attr:hidden:pp}).
Hence, here both NCM and PPM allow one to find this hidden attractor.

Around equilibrium $S_0$ we choose a small spherical vicinity
of radius $\delta$ (in our experiments we check $\delta \in [0.1,0.5]$)
and take $N$ random initial points on it (in our experiment $N = 4000$).
Using MATLAB, we integrate system \eqref{sys:lorenz-general} with these
initial points in order to explore the obtained trajectories.
We repeat this procedure several times in order to
get different initial points for trajectories on the sphere.
We get the following results: all the obtained trajectories
either attract to the stable equilibrium $S_+$
or the equilibrium $S_-$
and do not tend to the attractor.
This gives us a reason to classify the chaotic attractor, obtained in system
\eqref{sys:lorenz-general}, as the hidden one.

Remark that there exist hidden chaotic sets in the Rabinovich system,
which cannot be localized by PPM. For example,
for  parameters $r = 6.8$, $a = -0.5$, $b \in [0.99,\, 1]$
\cite{KuznetsovLMS-2016-INCAAM}, the hidden attractor obtained by NCM
is not localizable via PPM (see Fig.~\ref{fig:rabinovich:attr:hidden:no-pp}).

\begin{figure*}[ht]
 \centering
 \includegraphics[width=0.95\textwidth]{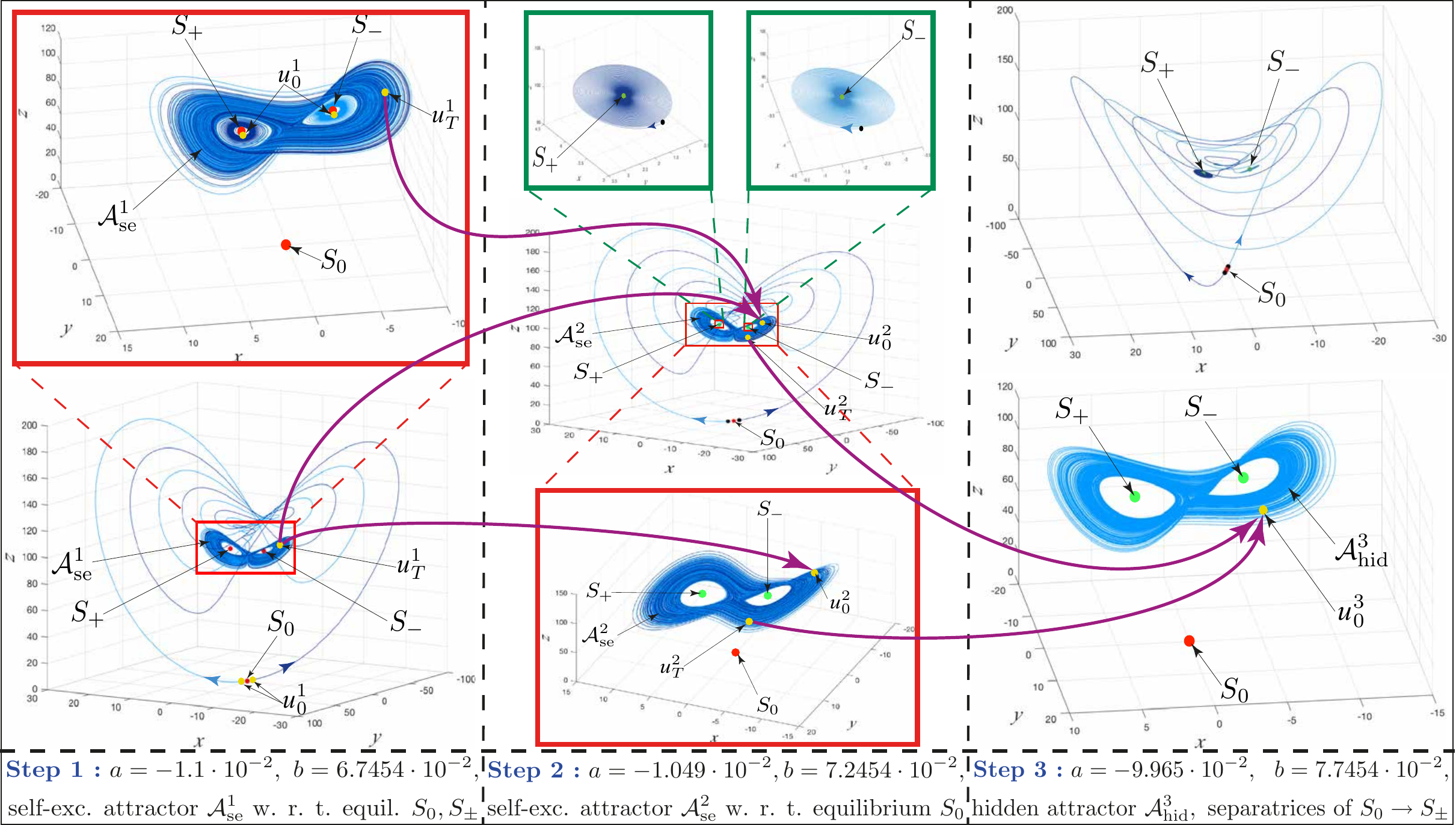}
 \caption{\label{fig:rabinovich:attr:hidden}
 Localization, by NCM, of a hidden attractor in system \eqref{sys:lorenz-general}
 with $r = 100$, $a = -9.965 \cdot 10^{-3}$, $b = 7.7454 \cdot 10^{-2}$.
 Trajectories $u^i(t) = (x^i(t), y^i(t), z^i(t)$ (blue) are defined on the time interval
 $[0, T]$, $T = 10^3$ and initial point (yellow)
 on $(i+1)$-th iteration is defined as $u_0^{i+1} := u_T^{i}$ (violet arrows),
 where $u_T^{i} = u^{i}(T)$ is a final point (yellow).}
\end{figure*}

\begin{figure}[ht]
 \centering
 \includegraphics[width=0.49\textwidth]{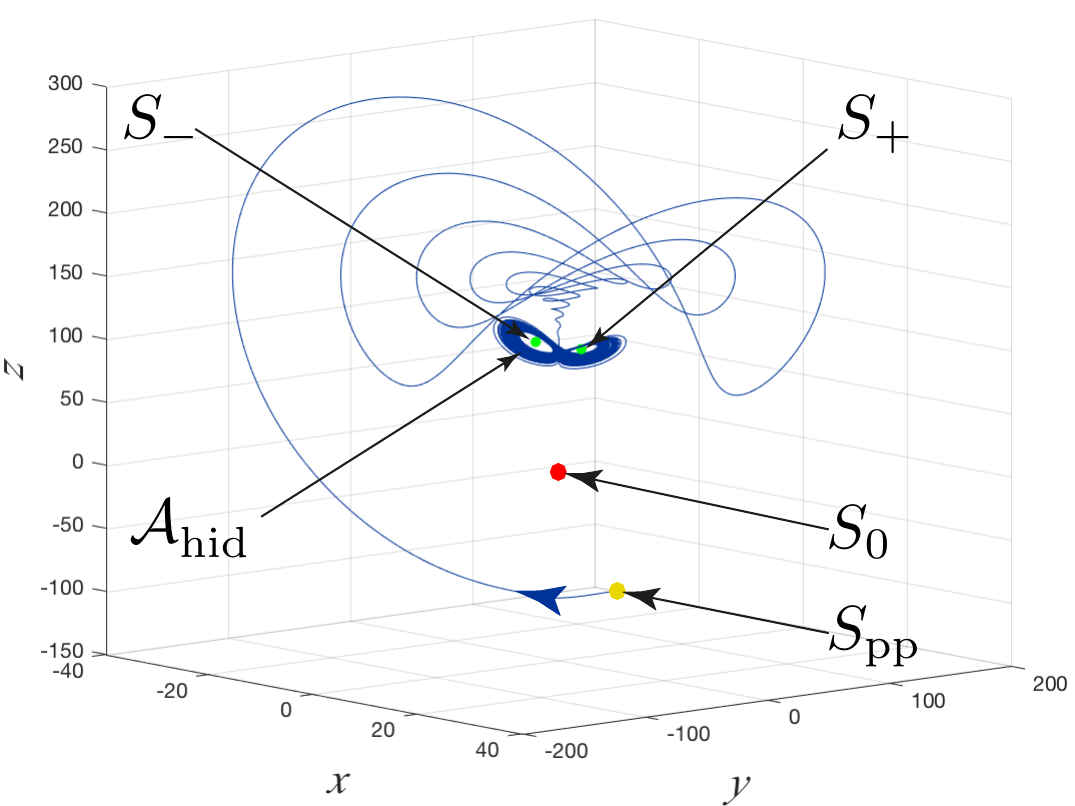}
 \caption{\label{fig:rabinovich:attr:hidden:pp}
 Localization of hidden attractor in system \eqref{sys:lorenz-general}
 with $r = 100$, $a = -9.965 \cdot 10^{-3}$, $b = 7.7454 \cdot 10^{-2}$
 from the perpetual point $S_{\rm pp} = (-0.2385, 49.1403, -101.4613)$.}
\end{figure}

\begin{figure}[ht]
 \centering
 \includegraphics[width=0.49\textwidth]{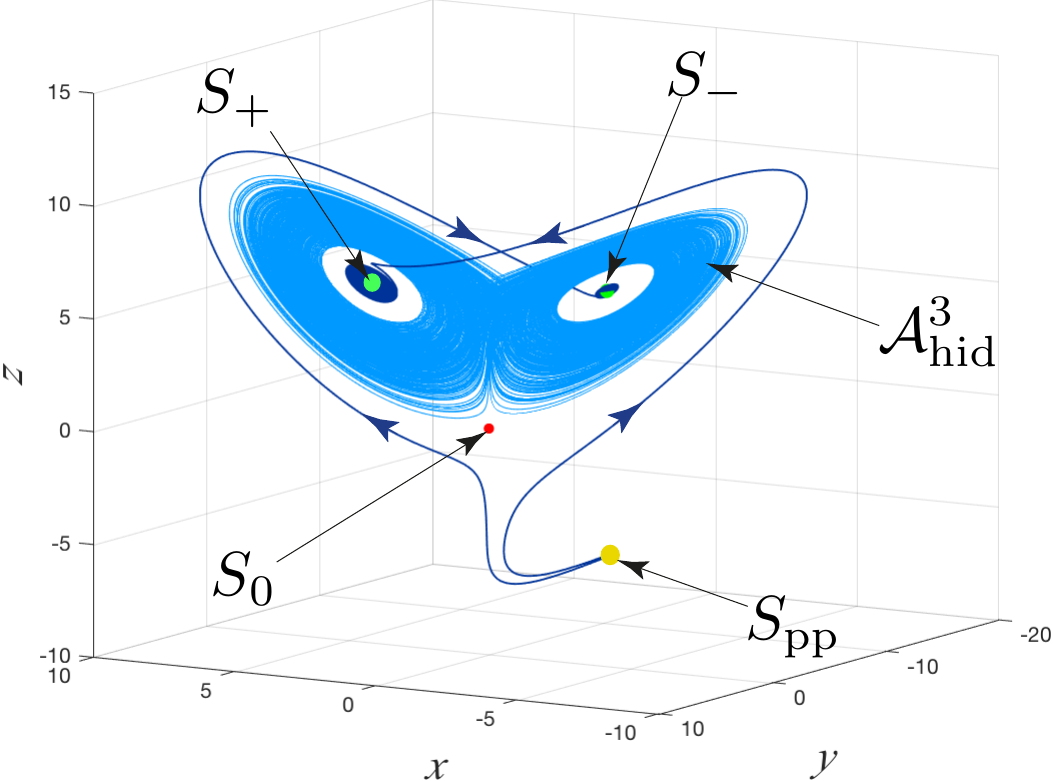}
 \medskip
 \caption{\label{fig:rabinovich:attr:hidden:no-pp}
 Hidden attractor in system \eqref{sys:lorenz-general} with
 $r = 6.8$, $a = -0.5$, $b = 0.99$ which can be localized via NCM
 (initial point ${u}_0 = (-0.1629,\, -0.2154, \, 1.9553)$),
 but cannot be localized from the perpetual point
 $S_{\rm pp} = (0.7431, -12.6109, -7.4424)$ (yellow).}
\end{figure}

\section{Computation of the Lyapunov dimension} \label{sec:lyapunov_dim}

\subsection{The Lyapunov dimension and Lyapunov exponents: finite-time and limit values.}

For the study of attractors
the Lyapunov exponents \cite{Lyapunov-1892} and Lyapunov dimension \cite{KaplanY-1979}
are found useful and have become widely spread
(see, e.g. \cite{GrassbergerP-1983,EckmannR-1985,ConstantinFT-1985,AbarbanelBST-1993,BoichenkoLR-2005}).
Since in numerical experiments we can consider only finite time,
in this paper we develop the concept of the \emph{finite-time Lyapunov dimension} \cite{Kuznetsov-2016-PLA}
and an approach to its reliable numerical computation.

Nowadays, various approaches to the Lyapunov dimension definition are used.
Here, we follow the definition
of \emph{finite-time Lyapunov dimension} from \cite{Kuznetsov-2016-PLA}
inspirited by the works of Douady \& Oesterl\'{e} \cite{DouadyO-1980}, Hunt~\cite{Hunt-1996},
and Rabinovich~et~al.~\cite{RabinovichEW-2000}.


Define by $u(t,u_0)$ a solution of system \eqref{sys:ode} such that $u(0,u_0)=u_0$,
and consider a map given by the evolutionary operator $\varphi^t(u_0) = u(t,u_0)$
(shift operator along a solution of \eqref{sys:ode}).
Since system \eqref{sys:ode} poses an absorbing set (see \eqref{absorb_set}),
the uniqueness and existence of solutions of system \eqref{sys:ode}
for $t \in [0,+\infty)$ take place and the system generates a \emph{dynamical system}
$\{\varphi^t\}_{t\geq0}$.
Let a nonempty closed bounded set $K \subset \mathbb{R}^3$
be invariant with respect to dynamical system generated by \eqref{sys:ode}
$\{\varphi^t\}_{t\geq0}$, i.e. $\varphi^t(K) = K$ for all $t \geq 0$
(e.g. $K$ is an attractor).
Further we use compact notations for
\emph{finite-time local Lyapunov dimension}:
$\dim_{\rm L}(t,u) = \dim_{\rm L}(\varphi^t, u)$,
the \emph{finite-time Lyapunov dimension}:
$\dim_{\rm L}(t,K) = \dim_{\rm L}(\varphi^t, K)$,
and for the \emph{Lyapunov dimension}
(or the Lyapunov dimension of dynamical system $\{\varphi^t\}_{t\geq0}$
with respect to $K$):
$\dim_{\rm L}K = \dim_{\rm L}(\{\varphi^t\}_{t \geq0}, K)$.

Consider linearization of system \eqref{sys:ode}
along the solution $\varphi^t(u)$:
\begin{equation} \label{sfl}
  \begin{aligned}
    & \dot v = J(\varphi^t(u))v,  \quad J(u) = Df(u),
  \end{aligned}
\end{equation}
where $J(u)$ is the $3\!\times\!3$ Jacobian matrix,
the elements of which are continuous functions of $u$.
Suppose that $\det J(u) \neq 0 \quad \forall u \in \mathbb{R}^3$.
Consider a fundamental matrix of solutions of linear system \eqref{sfl},
$D\varphi^t(u)$, such that $D\varphi^0(u) = I$, where $I$
is a unit $3\!\times\!3$ matrix.
Let $\sigma_i(t,u) = \sigma_i(D\varphi^t(u))$, $i = 1,2,3$,
be the singular values of $D\varphi^t(u)$
(i.e. $\sigma_i (t, u) > 0$ and ${\sigma_i (t, u)}^2$ are
the eigenvalues of the symmetric matrix $D\varphi^t(u)^*D\varphi^t(u)$
with respect to their algebraic multiplicity)\footnote{
Symbol $^*$ denotes the transposition of matrix.
},
ordered so that $\sigma_1(t,u) \geq \sigma_2(t,u) \geq \sigma_3(t,u) > 0$
for any $u$ and $t$.
A \emph{singular value function of order} $d \in [0,3]$ is defined as
\[
\begin{aligned}
  & \omega_d(D\varphi^t(u)) =
  \sigma_1(t, u)\cdots\sigma_{\lfloor d \rfloor}(t, u)
    \sigma_{\lfloor d \rfloor+1}(t, u)^{d-\lfloor d \rfloor},
  \\ &
  \omega_0(D\varphi^t(u)) = 1,
  \omega_3(D\varphi^t(u)) = \sigma_1(t, u)\sigma_2(t, u)\sigma_3(t, u),
\end{aligned}
\]
where ${\lfloor d \rfloor}$ is the largest integer less or equal to $d$.
For a certain moment of time $t$
\emph{finite-time local Lyapunov dimension} at the point $u$
is defined as \cite{Kuznetsov-2016-PLA}
\begin{equation}\label{locDOmaptmax}
  \dim_{\rm L}(t,u) = \max\{d \in [0,3]: \omega_{d}(D\varphi^t(u)) \geq 1 \}
\end{equation}
and the \emph{finite-time Lyapunov dimension}  of $K$
is defined as
\begin{equation}\label{DOmaptmax}
  \dim_{\rm L}(t, K) = \sup\limits_{u \in K} \dim_{\rm L}(t,u).
\end{equation}

The \emph{Douady--Oesterl\'{e} theorem} \cite{DouadyO-1980} implies that
for any fixed $t > 0$
the Lyapunov dimension of the map $\varphi^t$ with respect
to a closed bounded invariant set $K$, defined by \eqref{DOmaptmax},
is an upper estimate of the Hausdorff dimension of the set $K$:
$\dim_{\rm H}K \leq \dim_{\rm L}(t, K)$.

For the estimation of the Hausdorff dimension of invariant closed bounded set $K$
one can use the map $\varphi^t$ with any time $t$
(e.g. $t=0$ leads to the trivial estimate $\dim_{\rm H}K \leq 3$)
and, thus, the best estimation is
\(
  \dim_{\rm H}{K} \le \inf_{t\geq0}\dim_{\rm L} (t, K).
\)
The following property
\begin{equation}\label{DOlim}
  \inf_{t\geq0}\sup\limits_{u \in K} \dim_{\rm L}(t,u)
  = \liminf_{t \to +\infty}\sup\limits_{u \in K} \dim_{\rm L}(t,u)
\end{equation}
allows one to introduce the \emph{Lyapunov dimension}  of $K$
as \cite{Kuznetsov-2016-PLA}
\begin{equation}\label{DOinf}
  \dim_{\rm L} K 
  = \liminf_{t \to +\infty}\sup\limits_{u \in K} \dim_{\rm L}(t,u)
\end{equation}
and get an upper estimation of the Hausdorff dimension:
\[
  \dim_{\rm H}{K} \le \dim_{\rm L} K.
\]
Recall that a set with noninteger Hausdorff dimension
is referred to as a \emph{fractal set} \cite{EckmannR-1985}.

Consider a set of \emph{finite-time Lyapunov exponents}
at the point $u$:
\begin{equation}\label{ftLE}
  \LEs_i(t,u) = \frac{1}{t}\ln\sigma_i(t,u), \ t > 0, \quad i=1,2,3.
\end{equation}
Here the set $\{\LEs_i(t,u)\}_{i=1}^3$ is ordered by decreasing
(i.e. $\LEs_1(t,u) \geq \LEs_2(t,u) \geq \LEs_3(t,u)$ for all $t>0$).
Then for $j(t, u) = \lfloor \dim_{\rm L} (t, u) \rfloor <3$
and $s(t, u) = \dim_{\rm L} (t, u) - \lfloor \dim_{\rm L} (t, u) \rfloor$
 we have
\(
  0 = \frac{1}{t} \ln(\omega_{j(t,u)+s(t,u)}(D\varphi^t(u)))
    = \sum_{i=1}^{j(t,u)}\LEs_i(t,u) + s(t, u)\LEs_{j(t,u)+1}(t,u)
\)
and $j(t,u) = \max\{m: \sum_{i=1}^{m}\LEs_i(t,u) \geq 0\}$.
Thus, we get an analog of the \emph{Kaplan-Yorke formula} \cite{KaplanY-1979} with respect
to the set of finite-time Lyapunov exponents $\{\LEs_i(t,u)\}_{i=1}^3$ \cite{Kuznetsov-2016-PLA}:
\begin{multline}\label{lftKY}
   \!\!\!d_{\rm L}^{\rm KY}\!(\!\{\!\LEs_i(t,\!u)\!\}_{i=1}^3\!)\!=\!
   j(t,u)+\tfrac{\LEs_1(t,u) + \cdot\cdot + \LEs_{j(t,u)}(t,u)
   }{|\LEs_{j(t,u)\!+\!1}(t,u)|}
\end{multline}
which gives the finite-time local Lyapunov dimension:
\[
   \dim_{\rm L}(t, u) = d_{\rm L}^{\rm KY}(\{\LEs_i(t,u)\}_{i=1}^3).
\]
Thus, in the above approach the use of Kaplan-Yorke formula \eqref{lftKY} with
the finite-time Lyapunov exponents $\{\LEs_i(t,u)\}_{i=1}^3$
is rigorously justified by the Douady--Oesterl\'{e} theorem.

Note that the finite-time local Lyapunov dimension
is \emph{invariant under time scaling}:
$t \to at, a>0$ (e.g., it follows from \eqref{ftLE} and \eqref{lftKY}),
and the Lyapunov dimension is \emph{invariant under Lipschitz diffeomorphisms} \cite{KuznetsovAL-2016,Kuznetsov-2016-PLA}, i.e.
if the dynamical system $\{\varphi^t\}_{t\geq0}$ and
its closed bounded invariant set $K$
under a smooth change of coordinates $w = \chi(u)$
are transformed to the dynamical system
$\{\varphi_\chi^t\}_{t\geq0}$
and its closed bounded invariant set $\chi(K)$, respectively,
then
\(
  \dim_{\rm L}(\{\varphi^t\}_{t\geq0},K)
  =
  \dim_{\rm L}(\{\varphi_\chi^t\}_{t\geq0},\chi(K)).
\)


\subsection{Algorithm for numerical computation \\ of the finite-time Lyapunov dimension}

Applying the statistical physics approach and assuming the ergodicity
(see, e.g. \cite{KaplanY-1979,Ledrappier-1981,FredericksonKYY-1983,FarmerOY-1983}),
the Lyapunov dimension
of attractor $\dim_{\rm L} K$ 
is often estimated by the local Lyapunov dimension $\dim_{\rm L} (t, u_0)$,
corresponding to a
``typical'' trajectory, which belongs to the attractor:
$\{u(t,u_0), t \geq 0 \},\ u_0 \in K$,
and its limit value $\lim_{t\to+\infty}\dim_{\rm L} (t, u_0)$.
However, from a practical point of view,
the rigorous proof of ergodicity is a challenging task
\cite{BogoliubovK-1937,DellnitzJ-2002,Oseledec-1968,Ledrappier-1981}
and hardly it can effectively be done in a general case
(see, e.g. discussions in \cite{BarreiraS-2000}\cite[p.118]{ChaosBook}\cite{OttY-2008}\cite[p.9]{Young-2013}
\cite[p.19]{PikovskyP-2016}, and the works \cite{KuznetsovL-2005,LeonovK-2007}
on the \emph{Perron effects of the largest Lyapunov exponent sign reversals}).
An example of the rigorous use of the ergodic theory
for effective estimation of the Lyapunov dimensions can be found, e.g. in \cite{Schmeling-1998}.
In one of the pioneering works by Yorke~et~al.~\cite[p.190]{FredericksonKYY-1983}
the \emph{exact} limit values of finite-time Lyapunov exponents\footnote{
 \emph{The Lyapunov exponents} or LEs (see, e.g. \cite{Oseledec-1968})
characterize the rates of exponential growth of
the singular values of fundamental matrix of the linearized system.
The singular values correspond to the semiaxes of
n-dimensional ellipsoid, which is the image of the unit sphere by the linearized system.
See \cite{VallejoS-2017} for various related notions.
}
$\{\lim\limits_{t\to+\infty}\LEs_i(t,u)\}_{i}^3 = \{\LEs_i(u)\}$,
if they exist and are the same for all $u \in K$
$\big($i.e.
$ \{\LEs_i(u)\}_{i}^3 \equiv \{\LEs_i\}_{i}^3$ $\forall u \in K$
and
$\dim_{\rm L} K = d_{\rm L}^{\rm KY}(\{\LEs_i\}_{1}^3)
= j + \tfrac{\LEs_1 + \cdots + \LEs_{j}}{|\LEs_{j+1}|}$$\big)$,
are called the \emph{absolute} ones,
and it is noted that the \emph{absolute Lyapunov exponents} \emph{rarely exist}.
Note also that even if
a numerical approximation (visualization) $\widetilde{K}$ of the attractor $K$
is obtained, it is not straightforward how to get a point
on the attractor itself: $u \in K$.

Thus, a rather easy way to get reliable estimation of the Lyapunov dimension of attractor $K$
is to localize the attractor $K \subset K^{\varepsilon}$, to consider a grid of points $K^{\varepsilon}_{\rm grid}$ on $K^{\varepsilon}$, and
to find the maximum of the corresponding finite-time local Lyapunov dimensions for a certain time $t=T$:
$\max\limits_{u \in K^{\varepsilon}_{\rm grid}} \dim_{\rm L}(\varphi^T,u)
=\max\limits_{u \in K^{\varepsilon}_{\rm grid}} j(T,u) + \frac{\LEs_1(T,u) + \cdots +
\LEs_{j(T,u)}(T,u)}{|\LEs_{j(T,u)+1}(T,u)|}$.

Concerning $T$, remark that while the time series obtained from a \emph{physical experiment}
are assumed to be reliable on the whole considered time interval,
the time series, produced by the integration of \emph{mathematical dynamical model},
can be reliable on a limited time interval only\footnote{
In \cite{KehletL-2013,KehletL-2015} for the Lorenz system
the time interval of reliable computation
with 16~significant digits and error $10^{-4}$
is estimated
as $[0, 36]$, with error $10^{-8}$
is estimated as $[0, 26]$,
and reliable computation
for a longer time interval, e.g. $[0,10000]$ in \cite{LiaoW-2014}, is a challenging task.
Also if $u_0$ belongs to a \emph{transient chaotic set},
then $u(t,u_0)$ may have positive finite-time Lyapunov exponent
on a very large time interval, e.g. $[0,15000]$, but
finally $u(t,u_0)$ converges to a stable stationary point
as $t \to \infty$ and has nonpositive limit Lyapunov exponents.
}
due to computational errors,
and the closeness of the real trajectory
and the corresponding pseudo-trajectory calculated numerically
can be guaranteed on a limited short time interval only.
The computation of a pseudo-trajectory $\tilde u(t,u_0)$
on a longer time interval $t \in [0,T]$
often allows one to obtain a more complete visualization of
a chaotic attractor (pseudo-attractor)
due to computational errors
(caused by finite precision arithmetic and numerical integration of ODE)
and sensitivity to initial data.
However, 
for two long-time pseudo-trajectories
$\tilde u(t,u^1_0)$ and $\tilde u(t,u^2_0)$
the corresponding finite-time LEs can be, within the considered error,
similar due to averaging over time (see \eqref{ftLE})
and similar sets of points obtained
$\{\tilde u(t,u^1_0)\}_{t \geq0}$ and $\{\tilde u(t,u^2_0)\}_{t \geq0}$.
At the same time, the corresponding real trajectories $u(t,u^{1,2}_0)$
may have different LEs
(e.g. $u_0$ may correspond to an unstable periodic trajectory $u(t,u_0)$,
which is embedded in the attractor and does not allow one to visualize it).
Here one may recall 
the conjecture that the maximum of the local Lyapunov dimension
is achieved on a periodic orbit or a stationary point \cite[p.98]{Eden-1989-PhD}.
Also, if the trajectory belongs to a transient chaotic set, which can be (almost) indistinguishable
numerically from sustained chaos, then even very long-time computation may not reveal
the limit values of LEs
(see Figs.~\ref{fig:rab:LCE1:chaotic} and \ref{fig:rab:LCE1:transient}).

Thus, in general, the computation for a longer time does not imply a more precise
approximation of LEs.
Note, that there is no rigorous justification of the choice of $t$
and it is known that unexpected jumps
of $\dim_{\rm L}(t,K)$
can occur (see, e.g. Fig.~\ref{fig:tLD}).
Thus, it is reasonable to compute $\inf_{t \in [0,T)}\dim_{\rm L}(t,K)$
instead of $\dim_{\rm L}(T,K)$,
but, at the same time, for any $T$
the value $\dim_{\rm L}(T,K)$
gives also an upper estimate of $\dim_{\rm H}K$.

Finally, in the numerical experiments,
based on the \emph{finite-time Lyapunov dimension} definition \eqref{DOinf} from \cite{Kuznetsov-2016-PLA}
and the Douady--Oesterl\'{e} theorem \cite{DouadyO-1980},
we have
\begin{equation}\label{dimLmunest}
\begin{aligned}
  \dim_{\rm H}K
  \leq
  \dim_{\rm L}K
  \approx \inf_{t \in [0,T]} \max_{u \in K^{\varepsilon}_{\rm grid}} \dim_{\rm L}(t, u) \\
 =\inf_{t \in [0,T]}\max\limits_{u \in K^{\varepsilon}_{\rm grid}}
 \left(j(t,u) + \tfrac{\LEs_1(t,u) + \cdots +
 \LEs_{j(t,u)}(t,u)}{|\LEs_{j(t,u)+1}(t,u)|}
 \right)\\
  \leq
  \max_{u \in K^{\varepsilon}_{\rm grid}} \dim_{\rm L}(T, u) \approx \dim_{\rm L}(T,K).
\end{aligned}
\end{equation}

\subsection{Algorithm for numerical computation of the finite-time Lyapunov exponents}

Nowadays there are several widely used approaches
to numerical computation of the Lyapunov exponents, thus,
it is important \emph{to state clearly how the LEs being computed} \cite[p.121]{ChaosBook}.
Next we demonstrate the differences in the approaches.
To compute the finite-time Lyapunov exponents
one has to find the fundamental matrix $\Phi(t,u_0) = D\varphi^t(u_0)$
  of \eqref{sfl} from the following variational equation
  \begin{equation}\label{vareq}
    \!\begin{cases}
      \!\dot{u}(s,u_0)\!=\!f(u(s,u_0)), \ \ u(0,u_0) = u_0 \in U, \\
      \!\dot{\Phi}(s,u_0)\!=\!J(u(s,u_0)) \, \Phi(s,u_0),
       \ \Phi(0,u_0)\!=\!I,
    \end{cases}
    \!\!\!s \in [0,t],
  \end{equation}
  and its \emph{Singular Value Decomposition} (SVD)\footnote{
  See, e.g. implementation in MATLAB or GNU Octave (\url{https://octave-online.net}):
  {\ttfamily
  [U,S,V]=svd(A)}.
  }
  \[
    \Phi(t,u_0) \overset{\text{SVD}}{=}U(t,u_0){\rm \Sigma}(t,u_0){\rm V}^*(t,u_0),
  \]
  where $U(t,u_0)^*U(t,u_0) \equiv I \equiv {\rm V}(t,u_0)^*{\rm V}(t,u_0)$,
  ${\rm \Sigma}(t,u_0)=\text{\rm diag}\{\sigma_1(t,u_0),\sigma_2(t,u_0),\sigma_3(t,u_0)\}$
  is a diagonal matrix composed by the \emph{singular values} of $\Phi(t,u_0)$,
  and  compute the finite-time Lyapunov exponents $\{\LEs_i(t,u_0)\}_1^3$
  from ${\rm \Sigma}(t,\,u_0)$ as in \eqref{ftLE}.
  Further we also need the QR decomposition\footnote{
  For example, it can be done by the Gram-Schmidt orthogonalization procedure
  or the Householder transformation.
  MATLAB and GNU Octave (\url{https://octave-online.net})
  provide an implementation of the QR decomposition
  {\ttfamily [Q, R] = qr(A)}.
  To have matrix $R$ with positive diagonal elements, one can additionally use
  {\ttfamily Q = Q*diag(sign(diag(R))); R = R*diag(sign(diag(R)))}.
  See also \cite{RamasubramanianS-2000}.
  }
  \[
     \Phi(t,u_0) \overset{\text{QR}}{=} Q(t,u_0)R(t,u_0),
  \]
  where $R(t,u_0)$ is upper-triangular matrix with nonnegative diagonal elements
  $\{R[i,i]=R[i,i](t,u_0)\}_1^3$
  and $Q(t,u_0)^*Q(t,u_0) \equiv I$.

  To avoid the exponential growth of values in the computation,
  the time interval has to be represented as
  a union of sufficiently small intervals,
  e.g. $(0, T] =  (0,\tau]\cup(\tau,2\tau]\cdots \cup((k-1)\tau,k\tau=T]$.
  Then, using the cocycle property, the fundamental matrix can be represented as
  \begin{equation}\label{eq:mat_prod}
    \Phi(k \tau,u_0) = \Phi(\tau,u_{k-1}) \, \dots \,
    \Phi(\tau,u_{1}) \, \Phi(\tau,u_{0}).
  \end{equation}
  Here if $\Phi(m \tau,u_0)$ and $u_m=u(m \tau,u_0)$ are known,
  then $\Phi((m+1) \tau,u_0) = \Phi(\tau,u_{m}) \Phi(m \tau,u_0)$, where
  $\Phi(\tau,u_{m})$ is the solution of initial value problem \eqref{vareq}
  with $u(0) = u_m$ on the time interval $[0, \tau]$.

  By sequential QR decomposition of the product of matrices in \eqref{eq:mat_prod} we get
  \[
  \begin{aligned}
   & \Phi(k \tau,u_0) = \Phi(\tau,u_{k-1}) ..
    \Phi(\tau,u_{1}) \, \boxed{\Phi(\tau,u_{0})} =
    \\ &
    = \Phi(\tau,u_{k-1}) ..
    \boxed{\Phi(\tau,u_{1}) \, Q^0_1} \, R^0_1
    = .. \overset{\text{QR}}{=} \, \overbrace{Q^0_k}^{Q} \,\overbrace{R^0_k ..R^0_1}^{R}.
  \end{aligned}
  \]
  Then matrix with singular values
  $
   \Sigma(k\tau, u_0) =  U^*(k\tau,u_0) \,\Phi(k \tau,u_0) \, V(k\tau,u_0)
  $
   in the SVD can be approximated by sequential QR decomposition of the product of matrices:
  \[
  \begin{aligned}
    & \Sigma^0 = \Phi(k \tau,u_0)^* \, Q^0_k = (R^0_1)^* .. (R^0_k)^*
    \overset{\text{QR}}{=} \, Q^1_k \, R^1_k .. R^1_1,
    \\ &
    \Sigma^1 = (Q^0_k)^* \, \Phi(k \tau,u_0) \, Q^1_k =
    (R^1_1)^* .. (R^1_k)^*
    \overset{\text{QR}}{=} Q^2_k R^2_k .. R^2_1, \\
    & \dots
  \end{aligned}
  \]
  where
 \[
  \begin{aligned}
   & \Sigma^{j} = (R^{j}_1)^* .. (R^{j}_k)^*
  \!=\!\!\left(
    \begin{matrix}
      \sigma_1^{j} & 0 & 0 \\
      \cdot & \sigma_2^{j} & 0 \\
      \cdot & \cdot & \sigma_3^{j}
    \end{matrix}
    \right)\!\!
  \end{aligned}
 \]
 and \cite{RutishauserS-1963,Stewart-1997}
 \[
  \sigma_i^{j} = R^{j}_1[i,i]..R^{j}_k[i,i] \underset{j \to \infty}{\longrightarrow} \sigma_i(k\tau,u_0).
  \]
  Thus, the finite-time Lyapunov exponents can be approximated as
  \begin{equation}\label{LEapprox}
  \!\LEs_i(T, u_0)\!\approx\!
    \LEs_i^{j}(k\tau, u_0)\!=\!\tfrac{1}{t}\!\ln \sigma_i^{j}\!=\!
    \tfrac{1}{k\tau}\!\sum_{l = 1}^{k} \ln R_{l}^j[i,i].
  \end{equation}
  The MATLAB implementation of the above method
  for the computation of finite-time Lyapunov exponents
  with the fixed number of iterations $j$
  can be found, e.g., in \cite{LeonovKM-2015-EPJST}.
  For large $k$ the convergence can be very rapid:
  e.g. for the Lorenz system with the classical parameters
  ($r = 28$, $\sigma = 10$, $b = 8/3$, $a = 0$),
  $k = 1000$ and $\tau = 1$
  the number of approximations $j=1$ is taken in \cite[p.~44]{Stewart-1997}.
  For a more precise approximation of the finite-time Lyapunov exponents
  we can adaptively choose $j=j(l)$, $l\!=\!1,...,k$ so as to obtain a uniform estimate of
  \begin{equation}\label{LEapproxuni}
  \begin{aligned}
    \!\!\!\max_{i}|\LEs_i^{j\!-\!1}(l\tau, u_0)-\LEs_i^{j}(l\tau, u_0)|< \delta.
  \end{aligned}
  \end{equation}

  Remark that there is another widely used definition of the ``Lyapunov exponents''
  via the exponential growth rates of norms of the fundamental matrix columns
  $\big(v_1(t,u_0),v_2(t,u_0),v_3(t,u_0)\big)=\Phi(t,u_0)$:
   the \emph{finite-time Lyapunov characteristic exponents}
  $\{\LCEs_i(t,u_0)\}_{1}^3$
  are the set $\{\frac{1}{t}\ln||v^i(t,u_0)||\}_{1}^3$
  ordered by decreasing\footnote{
  To obtain all possible limit values of the
  finite-time Lyapunov characteristic exponents (LCEs) \cite{Lyapunov-1892}
  of a linear system ($\{\limsup_{t\to+\infty}\LCEs_i(t,u_0)\}_1^3$),
  one has to consider
  a \emph{normal fundamental matrix}, 
  whose sum of LCEs of columns is less or equal to
  the sum of LCEs of any other fundamental matrix \cite{Lyapunov-1892}.
}.
Benettin~et~al.~\cite{BenettinGGS-1980-Part2}
(Benettin's algorithm)
approximate the LCEs by \eqref{LEapprox} with $j=0$:
\begin{equation}\label{LCEapprox}
  \!\LCEs_i(k\tau, u_0)\!\approx\! \LEs_i^{0}(k\tau, u_0) =
  \frac{1}{k\tau}\!\sum_{l = 1}^{k} \ln R_{l}^0[i,i].
\end{equation}
The LCEs may differ from LEs,
thus, the corresponding Kaplan-Yorke formulas with respect to LEs and LCEs:
$\dim_{\rm L}(t, u_0) =d_{\rm L}^{\rm KY}(\{\LEs_i(t,u_0)\}_1^3)$ and
$d_{\rm L}^{\rm KY}(\{\LCEs_i(t,u_0)\}_1^3)$\footnote{
Rabinovich~et~al.~\cite[p.203,p.262]{RabinovichEW-2000}
refer this value as \emph{local dimension}
and note that it is a function of time and
may be different in different parts of the attractor.
},
may not coincide.
The following artificial analytical example
demonstrates the difference between LEs and LCEs.
The matrix \cite{Kuznetsov-2016-PLA,KuznetsovAL-2016}
\[
  R(t)\!=\!\left(\!\!
    \begin{array}{cc}
      1 & g(t)-g^{-1}(t)\\
      0 & 1 \\
    \end{array}
  \!\!\right), \ g(t)=\exp(\tfrac{t}{10})
\]
has the following ordered exact limit values
\[
\begin{aligned}
& \LCEs_1 =\lim\limits_{t \to +\infty} t^{-1} \ln g(t) = 0.1, \quad
\LCEs_2  = 0, \\
& \LEs_{1,2} =  \lim\limits_{t \to +\infty}t^{-1} \ln g^{\pm 1}(t) = \pm 0.1,
\end{aligned}
\]
where $\LCEs_2 \neq \LEs_2$.
 For the finite-time values we have
  \[\begin{aligned}
  & \LCEs_1(t)\!=\!\tfrac{1}{t} \ln\!\big((g(t)\!-\frac{1}{g(t)})^2+1\big)^{\tfrac{1}{2}}
  \!\in\!(0, 0.1],
  \LCEs_2(t)\!\equiv\!0, \\
  & \LEs_{1,2}(t) \equiv \LEs_{1,2} = \pm 0.1.
 \end{aligned}
 \]
Approximations by the above algorithm with $k=1$
are given in Table~\ref{table:LEs}.
 \begin{table}[ht]
  \centering
  \caption{Approximation of the finite-time Lyapunov exponents.}
  \begin{tabular}{
  |>{\centering}m{0.6cm}<{\centering}||
  >{\centering}m{2cm}<{\centering}|
  >{\centering}m{2cm}<{\centering}|
  >{\centering}m{2cm}<{\centering}|}
  \hline
  $j$ & $\LEs_{1,2}^{j}(5)$ & $\LEs_{1,2}^{j}(25)$ & $\LEs_{1,2}^{j}(100)$
  \tabularnewline\hhline{|=#=|=|=|}
  $0$ & $0$ & $0$ & $0$
  \tabularnewline\hline
  $1$ & $\pm 0.00797875$ & $\pm 0.04394912$ & $\pm 0.09360078$
  \tabularnewline\hline
  $2$ & $\pm 0.01585661$ & $\pm 0.07379280$ & $\pm 0.09986978$
  \tabularnewline\hline
  $3$ & $\pm 0.02353772$ & $\pm 0.08902280$ & $\pm 0.09999751$
  \tabularnewline\hline
  $4$ & $\pm 0.03093577$ & $\pm 0.09563887$ & $\pm 0.09999995$
  \tabularnewline\hline
  $5$ & $\pm 0.03797757$ & $\pm 0.09830568$ & $\pm 0.09999999$
  \tabularnewline\hline
  $10$ & $\pm 0.06638388$ & $\pm 0.09998593$ & $\pm 0.10000000$
  \tabularnewline\hline
  $50$ & $\pm 0.09993286$ & $\pm 0.09999999$ & $\pm 0.10000000$
  \tabularnewline\hline
  $100$ & $\pm 0.09999998$ & $\pm 0.09999999$ & $\pm 0.10000000$
  \tabularnewline\hline
  \end{tabular}
  \label{table:LEs}
  \end{table}

\noindent Remark that here the approximation of LCEs by Benettin's algorithm,
i.e. by \eqref{LCEapprox},
becomes worse with increasing time:
\begin{equation}\label{BenettinWrong}
 \begin{aligned}
  & \LCEs_1(t) \underset{t \to +0}{\longrightarrow} 0
  \equiv  \LEs_i^{0}(t) = \frac{1}{t} \ln 1 \equiv 0, \\
  & \LCEs_1(t) \underset{t \to +\infty}{\longrightarrow} 0.1
  \neq  \LEs_i^{0}(t) = \frac{1}{t} \ln 1 \equiv 0.
 \end{aligned}
\end{equation}
Thus, although relying on ergodicity, the notions of LCEs and LEs
often do not differ (see, e.g. Eckmann \& Ruelle \cite[p.620,p.650]{EckmannR-1985},
Wolf~et~al.~\cite[p.286,p.290-291]{WolfSSV-1985},
and Abarbanel~et~al.~\cite[p.1363,p.1364]{AbarbanelBST-1993}),
in the general case, the computations of LCEs by \eqref{LEapprox}
and LEs by \eqref{LCEapprox} may give non relevant results.
See also \cite[p.289]{BylovVGN-1966}, \cite[p.1083]{LeonovK-2007},
and numerical examples below.

\subsection{Estimation of the Lyapunov dimension \\ without integration of the system
and \\ the exact Lyapunov dimension}

While analytical computation of the Lyapunov exponents and Lyapunov dimension
is impossible in a general case, they can be estimated by
the eigenvalues of the symmetrized Jacobian matrix \cite{DouadyO-1980,Smith-1986}.
Let $\{\lambda_i(u_0)\}_{i=1}^{3}$ be the eigenvalues of the symmetrized Jacobian matrix
$\frac{1}{2} \left(J(u(t,u_0)) + J(u(t,u_0))^{*}\right)$,
ordered so that $\lambda_1(u_0) \ge \lambda_2(u_0) \ge \lambda_3(u_0)$.
The \emph{Kaplan-Yorke formula with respect to
the ordered set of eigenvalues of the symmetrized Jacobian matrix} \cite{Kuznetsov-2016-PLA}
gives an upper estimation of the Lyapunov dimension:
\(
 \dim_{\rm L}K \leq
 \sup_{u \in K}d_{\rm L}^{\rm KY}\big(\{\lambda_{j}(u_0)\}_{i=1}^3\big)
\).
In the general case,
one cannot get the same values of $\{\lambda_{j}(u_0)\}_{i=1}^3$
at different points $u_0$,
thus, the maximum of $d_{\rm L}^{\rm KY}\big(\{\lambda_{j}(u_0)\}_{i}^3\big)$
on $K$ has to be computed.
To avoid numerical localization of the set $K$,
we can consider an analytical localization,
e.g. by the absorbing set $\mathcal{B} \supset K$.
Thus, for the corresponding grid of points $\mathcal{B}_{\rm grid}$
we expect in numerical experiments the following
\begin{multline}\label{dLKYeigcomp}
 \dim_{\rm H}K \leq \dim_{\rm L}K
 \leq \sup_{u_0 \in K} d_{\rm L}^{\rm KY}\big(\{\lambda_{j}(u_0)\}_{1}^3\big) \leq  \\
 \sup_{u_0 \in \mathcal{B}}\!\! d_{\rm L}^{\rm KY}\big(\{\lambda_{j}(u_0)\}_{1}^3\big)
 \!\approx\!\!
 \max_{u_0 \in \mathcal{B}_{\rm grid}} \!\!
 j(u_0) + \tfrac{\lambda_{1}(u_0) + \cdot\cdot + \lambda_{j(u_0)}(u_0)}
 {|\lambda_{j(u_0)+1}(u_0)|}.
\end{multline}

If the Jacobian matrix $J(u_{eq})$ at one of the equilibria has simple real eigenvalues:
$\{\lambda_i(u_{eq})\}_{i=1}^3$, $\lambda_{i}(u_{eq}) \geq \lambda_{i+1}(u_{eq})$,
then \cite{Kuznetsov-2016-PLA} the invariance of the Lyapunov dimension
with respect to linear change of variables implies
\begin{equation}\label{dimLeq}
  \dim_{\rm L}u_{eq}
  =
  d_{\rm L}^{\rm KY}(\{\lambda_i(u_{eq})\}_{i=1}^3).
\end{equation}
If the maximum of local Lyapunov dimensions on the global attractors,
which involves all equilibria, is achieved at an equilibrium point:
$\dim_{\rm L} u^{cr}_{eq} = \max_{u_0 \in K} \dim_{\rm L} u_0$,
then this allows one to get analytical formula of
the \emph{exact Lyapunov dimension}\footnote{
This term was suggested by Doering~et~al.~in~\cite{DoeringGHN-1987}.}.
In general, a \emph{conjecture on the Lyapunov dimension of self-excited attractor} \cite{Kuznetsov-2016-PLA,KuznetsovL-2016-ArXiv}
is that for a typical system
the Lyapunov dimension of a self-excited attractor
does not exceed the Lyapunov dimension of one of unstable equilibria,
the unstable manifold of which intersects with the basin of attraction
and visualize the attractor.

To avoid numerical computation of the eigenvalues,
one can use an effective analytical approach
\cite{Leonov-1991-Vest,LeonovB-1992,Kuznetsov-2016-PLA},
which is based on a combination of the Douady-Oesterl\'{e}
approach with the direct Lyapunov method:
for example, in \cite{LeonovB-1992} for
system \eqref{sys:lorenz-general} with $b = 1$
it is analytically obtained the following estimate 
\[
  \dim_{\rm L} K \leq
  3 - \frac{2 (\sigma + 2)}{\sigma + 1 + \sqrt{(\sigma - 1)^2 + \frac{16\,r}{3\,\sigma}}}.
\]
The proof of the above conjecture and
analytical derivation of the exact Lyapunov dimension formula
for system \eqref{sys:lorenz-general} is an open problem.

In \cite{BoichenkoL-1998,PogromskyM-2011} it is demonstrated how a technique similar to the above can be effectively used to derive constructive upper bounds of the topological entropy of dynamical systems.

\begin{figure}[t]
    \centering
    \includegraphics[width=0.49\textwidth]{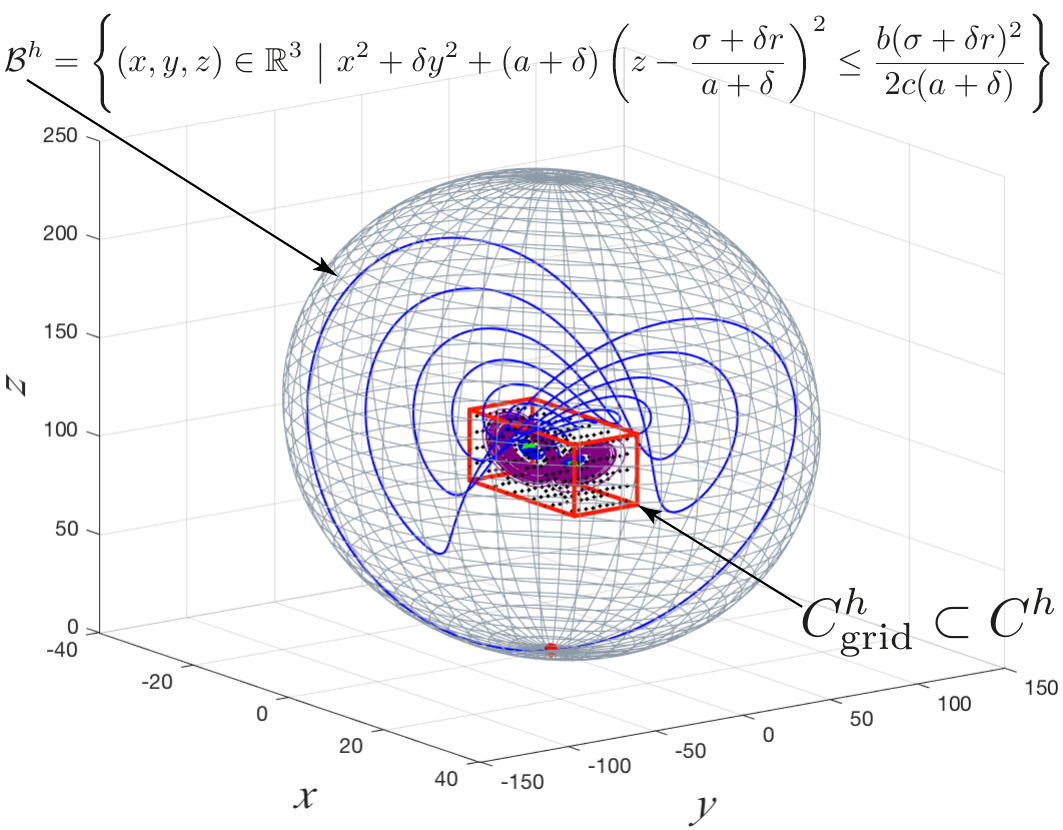}
    \caption{Localization of the hidden attractor of system \eqref{sys:lorenz-general}
    with $r = 100$, $a = -9.965 \cdot 10^{-3}$, $b = 7.7454 \cdot 10^{-2}$
    by the absorbing set
    $\mathcal{B}^{h}$ with $\delta = -a + 0.1$,
    cuboid $C^{h} =[-11,11]\times[-17,19]\times[80,117]$,
    and the corresponding grid of points $C^{h}_{\rm grid}$.
    }
    \label{fig:rab:grid}
\end{figure}

\section{The finite-time Lyapunov dimension in the case of hidden attractor and multistability}
Consider the dynamical system $\{\varphi^t\}_{t\geq0}$
generated by system \eqref{sys:lorenz-general}
with parameters \eqref{eq:params-relation} and
its attractor $K$. 
Here $\varphi^t\big((x_0,y_0,z_0)\big)$
is a solution of \eqref{sys:lorenz-general}
with the initial datum  $(x_0,y_0,z_0)$.
Since the dynamical system $\{\varphi^t_{\rm R}\}_{t\geq0}$,
generated by the Rabinovich system \eqref{sys:rabinovich},
can be obtained
from $\{\varphi^t\}_{t\geq0}$ by the smooth transformation $\chi^{-1}$,
inverse to \eqref{xyzchange}, and inverse rescaling time \eqref{timechange}
$t \to \nu_1 t$,
we have
$\dim_{\rm L}(\{\varphi^t\}_{t\geq0},K)=\dim_{\rm L}(\{\varphi^t_{\rm R}\}_{t\geq0},\chi^{-1}(K))$.
In our experiments, we consider system \eqref{sys:lorenz-general}
with parameters $r = 100$, $a = -9.965 \cdot 10^{-3}$, $b = 7.7454 \cdot 10^{-2}$
corresponding to the hidden chaotic attractor.

\begin{table*}[ht]
\centering
\caption{
Numerical estimation of finite-time Lyapunov dimension
in the case of hidden attractor (see Fig.~\ref{fig:rabinovich:attr:hidden})}
\begin{tabular}{
|>{\centering}m{3.4cm}<{\centering}||
>{\centering}m{3.3cm}<{\centering}|
>{\centering}m{5cm}<{\centering}|
>{\centering}m{2cm}<{\centering}|
>{\centering}m{2.3cm}<{\centering}|}
\hline
& $t = 100$ \\ $u = (0, \, 1, \, 98)$ & $t = 100$ \\ $u = (0.0099, \, 0.0995, \, 0)$ &
$t = 100$ \\ $\max_{u \in C^{h}_{\rm grid}}$ & $\displaystyle \inf_{t \in [0,\, 100]} \max_{u \in C^{h}_{\rm grid}}$
\tabularnewline\hhline{|=#=|=|=|=|}
$\dim_{\rm L}(t, \, u)$ = ${\small {d}_{\rm L}^{\rm KY}(\{{\rm LE}_i(t, \, u)\}_{i=1}^3)}$ \\
 & $2.1474$ & $1.4987$ & $2.2063$ & $2.2050$
\tabularnewline\hline
${\small {d}_{\rm L}^{\rm KY}(\{{\rm LCE}_i(t, \, u)\}_{i=1}^3)}$ \\
 & $2.1338$ & $1.1213$ & $2.2105$ & $2.2076$
\tabularnewline\hline
\end{tabular}
\label{table:results:hid}
\end{table*}

\begin{figure*}[ht]
  \centering
    \includegraphics[width=\textwidth]{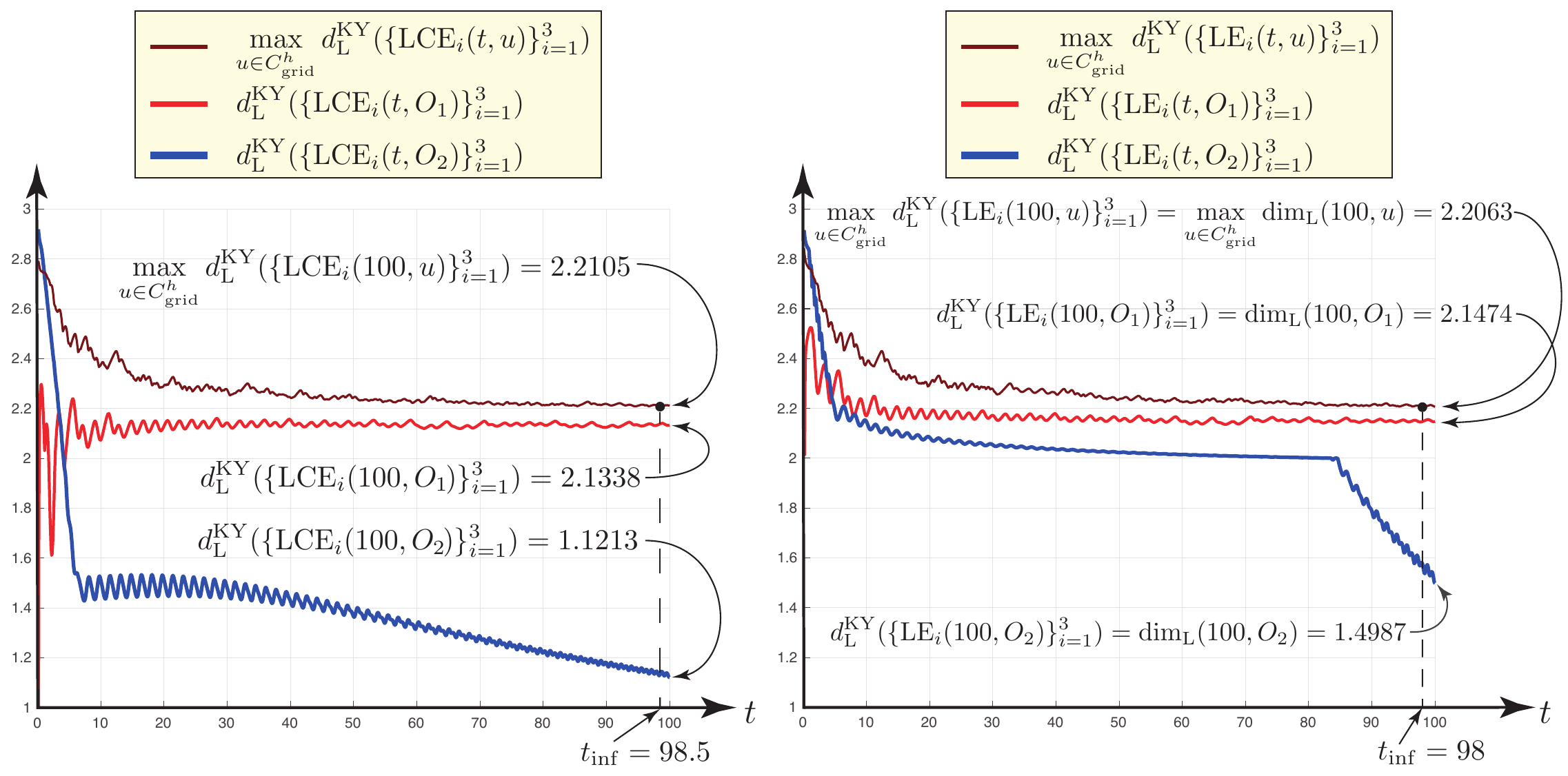}
  \caption{
    Dynamics of the finite-time local Lyapunov dimensions estimation
    on the time interval $t \in \lbrack 0, 100 \rbrack$: \\
    the maximum on the grid of points (dark red),
    at the point $O_1 = (0, \, 1, \, 98) \in C^{h}_{\rm grid}$ (light red), \\
    at the point $O_2 = (0.0099, \, 0.0995, \, 0)$
    from the one-dimensional unstable manifold of $S_0$ (blue).
  }
  \label{fig:tLD}
\end{figure*}

In Fig.~\ref{fig:rab:grid} it is shown the grid of points $C^{h}_{\rm grid}$
filling the hidden attractor:
the grid of points fills cuboid $C^{h} = [-11,11]\times[-17,19]\times[80,117]$
with the distance between points equals to $0.5$.
The time interval is $[0,\, T=100]$, $k=1000$, $\tau=0.1$,
and the integration method is MATLAB ode45 with predefined parameters.
The infimum on the time interval
is computed at the points $\{t_k\}_{1}^{N}$ with time step $\tau=t_{i+1}-t_i=0.1$.
Note that if for a certain time $t=t_k$ the computed trajectory is out of the cuboid,
the corresponding value of finite-time local Lyapunov dimension
is not taken into account in the computation of maximum
of the finite-time local Lyapunov dimension
(e.g. there are trajectories with initial data in cuboid,
which are attracted to the zero equilibria, i.e. belong to its stable manifold,
e.g. system \eqref{sys:lorenz-general} with $x=y=0$ is $\dot z = -bz$).
For the finite-time Lyapunov exponents (FTLEs) computation
we use MATLAB realization from \cite{LeonovKM-2015-EPJST}
based on \eqref{LEapprox} with $j=2$.
For computation of the finite-time Lyapunov characteristic exponents (FTLCEs)
we use MATLAB realization from \cite{KuznetsovMV-2014-CNSNS} based on \eqref{LCEapprox}.
For the considered set of parameters we compute:
\begin{enumerate}[label=(\roman*)]
  \item finite-time local Lyapunov dimensions $\dim_{\rm L}(100,\cdot)$
      at the point $O_1 = (0, \, 1, \, 98)$, which belongs to the grid $C^{h}_{\rm grid}$,
      and at the point $O_2 = (0.0099, \, 0.0995, \, 0)$ on the unstable manifold of zero equilibrium $S_0$;
  \item maximum of the finite-time local Lyapunov dimensions at the points of grid,
      $\max_{u \in C^{h}_{\rm grid}} \dim_{\rm L}(t,u)$,
      for the time points $t=t_k=0.1\,k$ $(k=1,..,1000)$;
  \item the corresponding values, given by the Kaplan-Yorke formula
      with respect to finite-time Lyapunov characteristic exponents.
\end{enumerate}


The results are given in Table~\ref{table:results:hid}.
The dynamics of finite-time local Lyapunov dimensions
for different points and their maximums on a grid of points
are shown in Fig.~\ref{fig:tLD}.

For the absorbing set $\mathcal{B}^{h}$ and the corresponding
grid of points $\mathcal{B}^{h}_{\rm grid}$
(the distance between grid points is 5),
by estimation \eqref{dLKYeigcomp} 
we get the following estimate:
\begin{multline}\label{KYeigonabs-h}
 \dim_{\rm H}K \leq
 \dim_{\rm L} K \leq
 \sup_{u \in \mathcal{B}^{h}}d_{\rm L}^{\rm KY}\big(\{\lambda_{j}(u)\}_{i=1}^3\big) \\
 \approx
 \sup_{u \in \mathcal{B}^{h}_{\rm grid}}d_{\rm L}^{\rm KY}\big(\{\lambda_{j}(u)\}_{i=1}^3\big)
 = 2.97001... \, .
\end{multline}

Assuming $\sigma + 1 \geq b$, the eigenvalues of the unstable zero equilibrium $S_0$
\[
  \lambda_{1,3}(S_0) = -\tfrac{1}{2} \left[(\sigma + 1)\!\mp\!\sqrt{(\sigma - 1)^2 + 4 \sigma r}\right]\!,\,\lambda_{2}(S_0) = - b
\]
 have the following order
$\lambda_{1}(S_0) > \lambda_{2}(S_0) \geq \lambda_{3}(S_0)$ and
by \eqref{dimLeq}
for the considered values of parameters 
we get
\begin{multline}\label{dimS0-val}
  \dim_{\rm L} S_0 =  d_{\rm L}^{\rm KY}(\{\lambda_i(S_0)\}_{i=1}^3)
  = 2 + \tfrac{\lambda_{1}(S_0) + \lambda_{2}(S_0)}{|\lambda_{3}(S_0)|} \\
    \\
  = 3 - \tfrac{2 (\sigma + b + 1) }{(\sigma + 1) + \sqrt{(\sigma - 1)^2 + 4 \sigma r}}
  =2.8111...\ .
\end{multline}

\begin{figure*}[t]
 \centering
 \subfloat[{Trajectory $u(t,u_{\rm init})$ for $t \in [0,~T_1]$, \, $T_1 \approx 15295$.}]{
    \label{fig:rab:attr:chaotic}
    \includegraphics[width=0.425\textwidth]{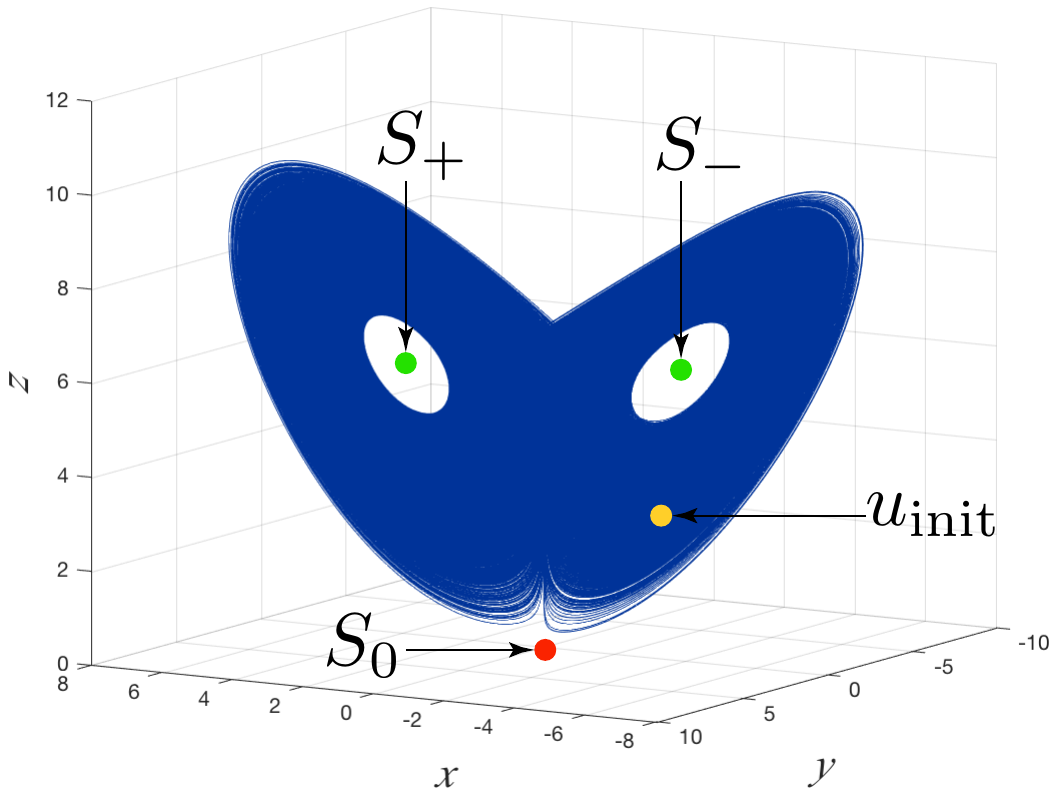}
  }\quad
  \subfloat[{Trajectory $u(t,u_{\rm init})$ for $t \in [0,~T]$, \, $T = 500000$}]{
    \label{fig:rab:attr:transient}
    \includegraphics[width=0.425\textwidth]{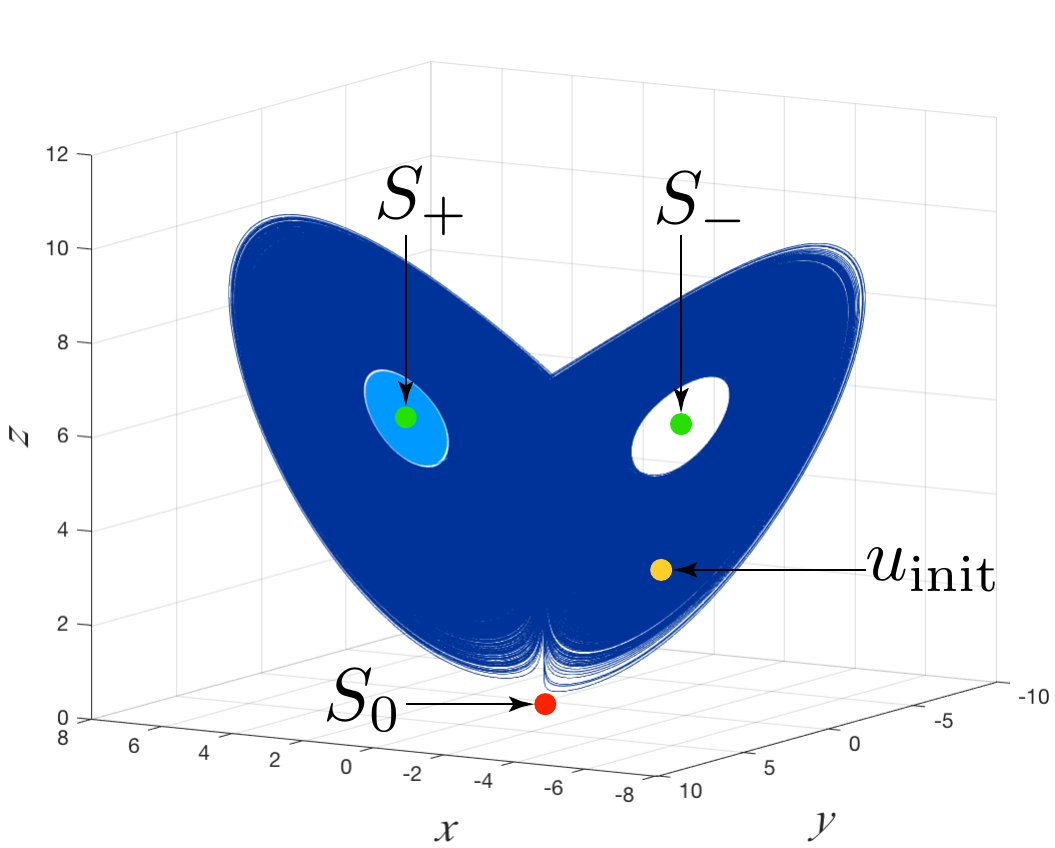}
  }
\caption{
The trajectory forms a chaotic set, which looks like an ``\emph{attractor}'' (navy blue)
and then 
tends to $S_+$ (cyan).}
\end{figure*}
\begin{figure*}[t]
 \centering
 \subfloat[{$\LCEs_1(t, \, u_{\rm init})$, $t \in [0,~T_1]$, \, $T_1 \approx 15295$.}]{
    \label{fig:rab:LCE1:chaotic}
    \includegraphics[width=0.425\textwidth]{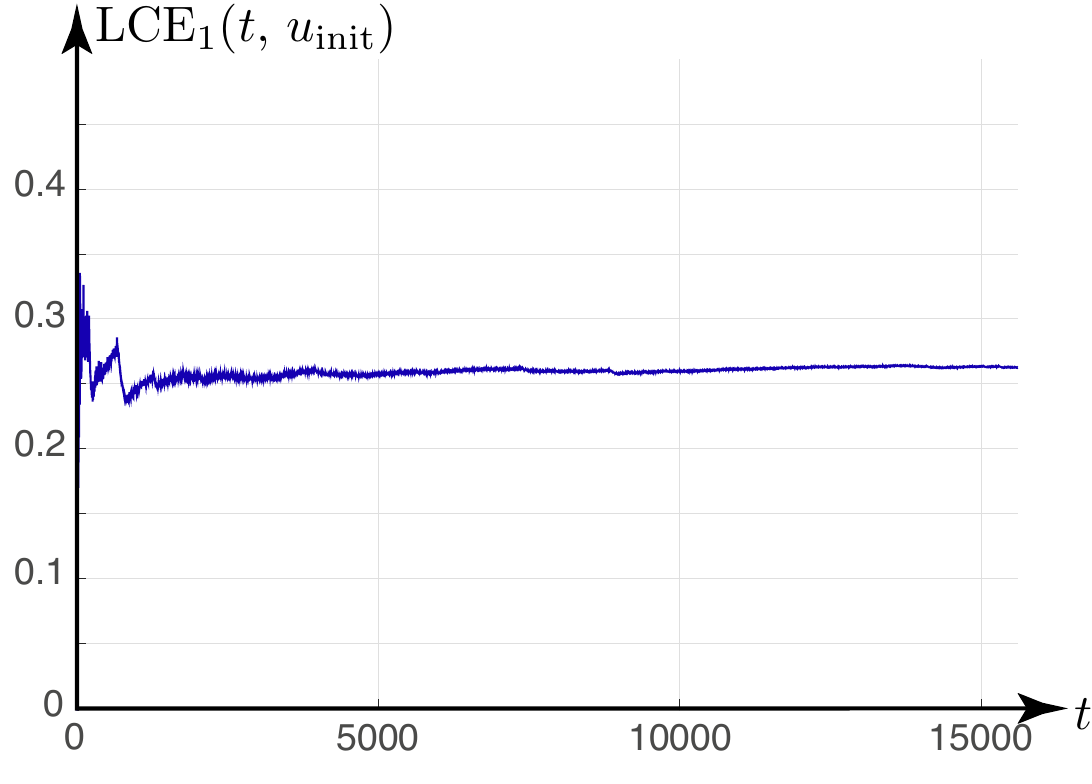}
  }\quad
 \subfloat[{${d}_{\rm L}^{\rm KY}(\{\LCEs_i(t, \, u_{\rm init})\}_{i=1}^3)$,
  $t \in [0,~T_1]$, \, $T_1 \approx 15295$.}]{
    \label{fig:rab:LD:chaotic}
    \includegraphics[width=0.425\textwidth]{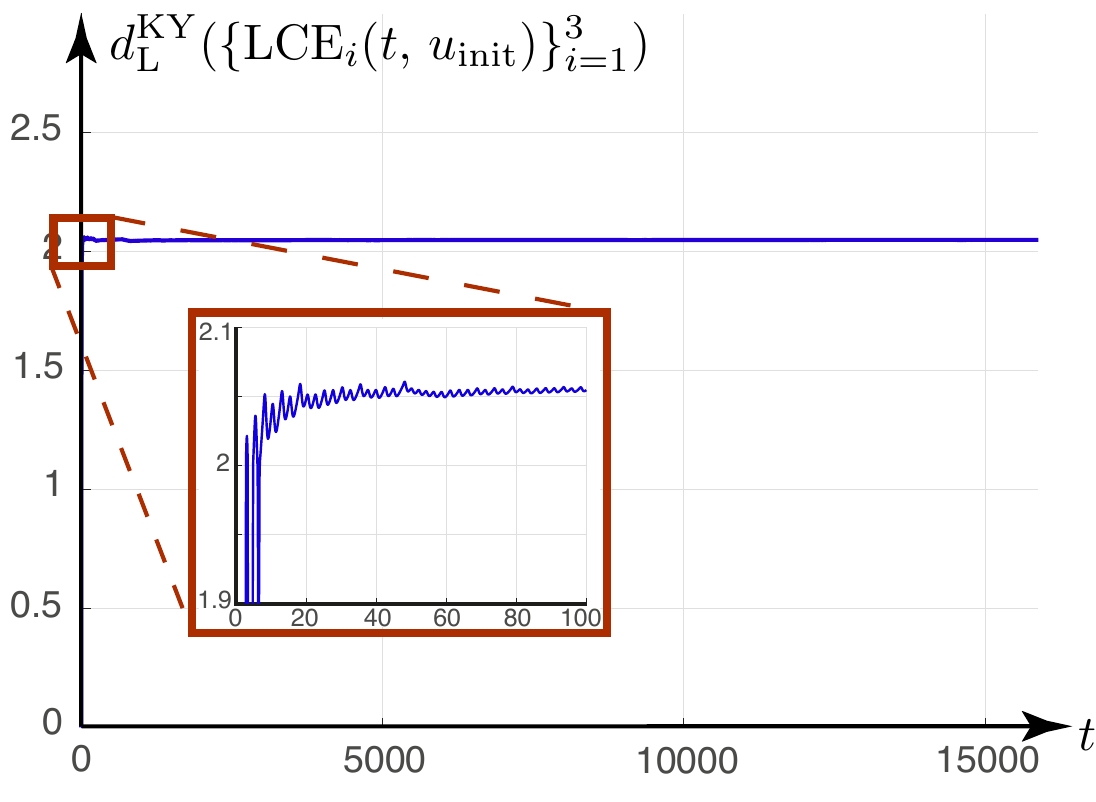}
  }
\caption{
  Numerical computation of $\LCEs_1(t, \, u_{\rm init})$
  and  ${d}_{\rm L}^{\rm KY}(\{\LCEs_i(t, \, u_{\rm init})\}_{i=1}^3)$
  for the time interval $[0, T_1\approx 15295]$.
}
\end{figure*}
\begin{figure*}[!ht]
 \centering
 \subfloat[{$\LCEs_1(t, \, u_{\rm init})$, $t \in [0,~T]$, \, $T = 5 \cdot 10^5$.}]{
    \label{fig:rab:LCE1:transient}
    \includegraphics[width=0.425\textwidth]{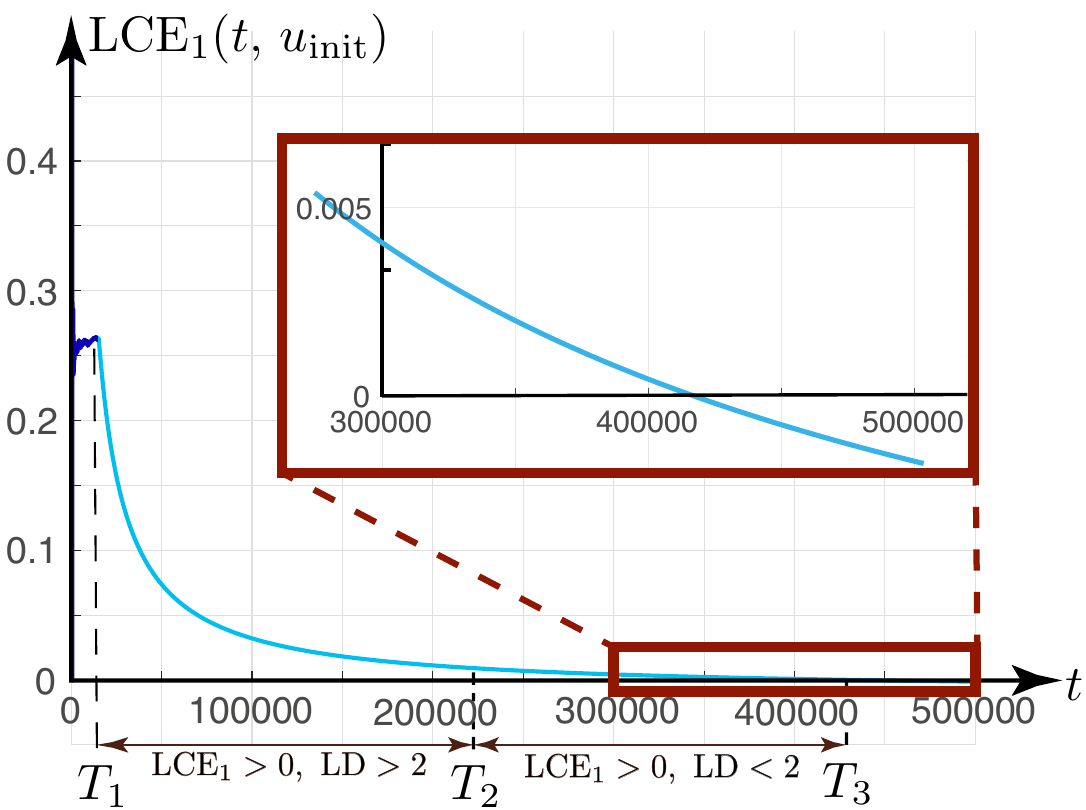}
  }\quad
 \subfloat[{${d}_{\rm L}^{\rm KY}(\{\LCEs_i(t, \, u_{\rm init})\}_{i=1}^3)$,
  $t \in [0,~T]$, \, $T = 5 \cdot 10^5$.}]{
    \label{fig:rab:LD:transient}
    \includegraphics[width=0.425\textwidth]{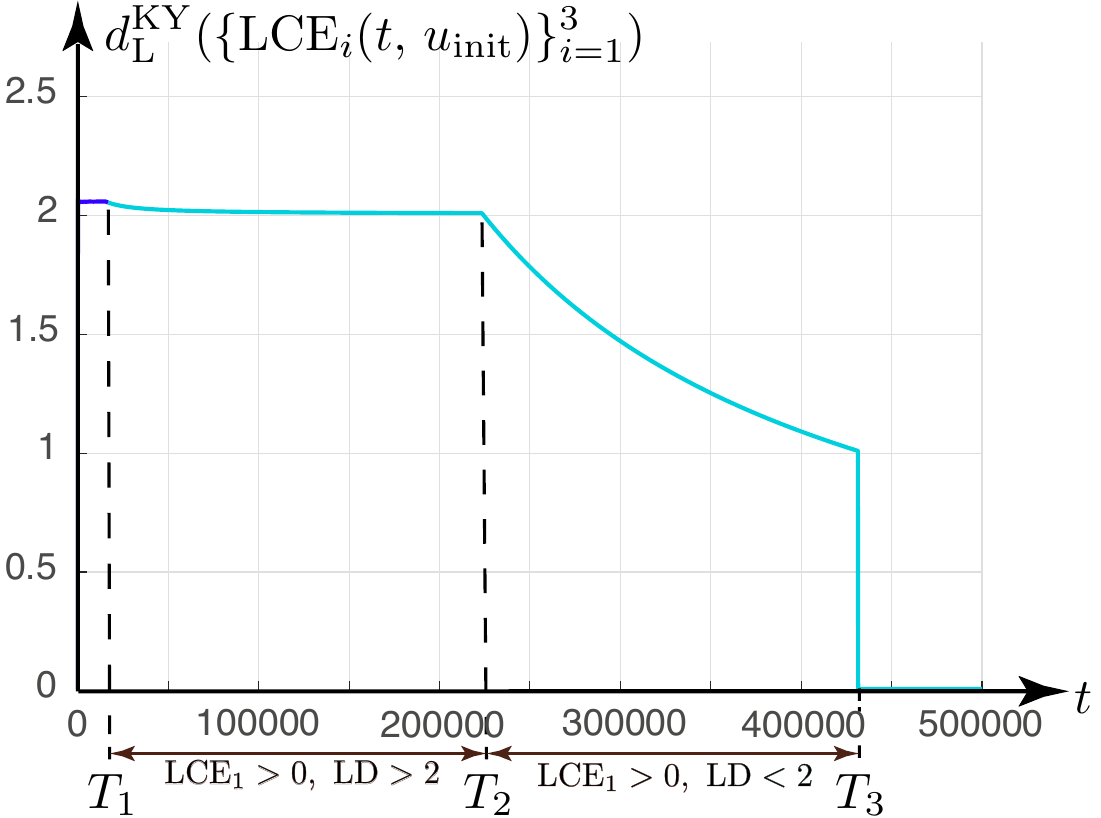}
  }
\caption{Numerical computation of $\LCEs_1(t, \, u_{\rm init})$
    and ${d}_{\rm L}^{\rm KY}(\{\LCEs_i(t, \, u_{\rm init})\}_{i=1}^3)$
    for the time interval $[0,\, 5 \cdot 10^5]$.
  }
\end{figure*}

The above numerical experiments lead to the following important remarks.
While the Lyapunov dimension,
unlike the Hausdorff dimension,
is not a dimension in the rigorous sense \cite{HurewiczW-1941} 
(e.g. the Lyapunov dimension of the saddle point $S_0$ in \eqref{dimS0-val} is noninteger),
it gives an upper estimate of the Hausdorff dimension.
If the attractor $K$ or the corresponding absorbing set $\mathcal{B} \supset K$
is known (see, e.g. \eqref{absorb_set})
and the purpose is to demonstrate that $\dim_{\rm H} K \leq \dim_{\rm L} K < 3$,
then it can be achieved
without integration of the considered dynamical system
(see, e.g. \eqref{KYeigonabs-h}).
If the purpose is to get a precise estimation of the Hausdorff dimension,
then one can use \eqref{dimLmunest} and has to compute
the finite-time Lyapunov dimension,
i.e. to find the maximum of the finite-time local Lyapunov dimensions
on a grid of points for a certain time.
To be able to repeat a computation of finite-time Lyapunov dimension,
one need to know the initial points of considered trajectories
$\{u_i\}_{i=1}^{N}=K_{\rm grid}$ on the set $K$,
time interval $(0,T] = \bigcup_{i=0}^{M-1} (t_i,t_{i+1}]$,
and the method of the finite-time Lyapunov estimation.

\section{Computation of the Lyapunov dimension and transient chaos} \label{sec:transient}

Now we consider an example, which demonstrates
difficulties in the reliable numerical computation of the Lyapunov dimension
(i.e. numerical approximation of the limit value of the finite-time Lyapunov dimension).

Consider system \eqref{sys:lorenz-general} with parameters
$r = 6.485$, $a = -0.5$, $b = 0.85$ for which equilibrium $S_0$ is a saddle point
and equilibria $S_\pm$ are stable focus-nodes.
We integrate numerically\footnote{
Our experiment was carried out on the
2.5 GHz Intel Core i7 MacBook Pro laptop,
for numerical integration we use MATLAB~R2016b.
To simplify the repetition of results,
we use the single-step fifth-order Runge-Kutta method {\ttfamily ode5} from
\url{https://www.mathworks.com/matlabcentral/answers/98293-is-there-a-fixed-step-ordinary-differential-
equation-ode-solver-in-matlab-8-0-r2012b\#answer\\107643}.
Corresponding numerical simulation of the considered trajectory
can be performed using the following code:
{\ttfamily phiT\_u = feval('ode5', @(t, u)
[-3.2425*u(1) + 3.2425*u(2) + 0.5*u(2)*u(3);
              6.485*u(1) - u(2) - u(1)*u(3);
                u(1)*u(2) - 0.85*u(3)],
  0 : 0.01 : 15295, [-2.089862710574761, -2.500837780529156, 2.776106323157132]);}
} the trajectory with initial data $u_{0} = (0.5,\,0.5,\,0.5)$ in the
vicinity of the $S_0$.
We discard the part of the trajectory, corresponding to the initial transition process
(for $[0,t_{\rm tp} = 25000]$),
and get the point $u_{\rm init} = ($ $-2.089862710574761$, $-2.500837780529156$, $2.776106323157132).$
Further, we numerically approximate the finite-time Lyapunov exponents
and dimension for the time interval $[0,~T]$
by Benettin's algorithm (see approximation \eqref{LCEapprox}
and MATLAB code in \citep{LET-1998}).


The trajectory computed on the time interval $[0, T_1\approx 15295]$
traces a \emph{chaotic set} in the phase space,
which looks like an ``\emph{attractor}'' (see Fig.~\ref{fig:rab:attr:chaotic}).
Further integration with $t > T_1$ leads to the collapse of the ``\emph{attractor}''
(see Fig.~\ref{fig:rab:attr:transient}),
i.e. the ``\emph{attractor}'' turns out to be a \emph{transient chaotic set}.
However on the time interval $t \in [0,~T_3\approx431560]$
we have $\LCEs_1(t, \, u_{\rm init}) > 0$ (see Fig.~\ref{fig:rab:LCE1:transient})
and, thus, one may conclude that the behavior is chaotic,
and for the time interval $t \in [0,~T_2\approx 223447]$
we have ${d}_{\rm L}^{\rm KY}(\{\LCEs_i(t, \, u_{\rm init})\}_{i=1}^3) > 2$
(see Fig.~\ref{fig:rab:LD:transient}).
This effect is due to the fact that the finite-time Lyapunov exponents and
finite-time Lyapunov dimension are averaged values over the considered time interval
and, therefore, may reflect a change in the qualitative behavior of the trajectory
with a delay.
Since the lifetime of transient chaotic process can be extremely long
and taking into account the limitations of reliable integration of chaotic ODEs,
the long-time computation of the finite-time Lyapunov exponents
and the finite-time Lyapunov dimension
does not necessary lead to a more relevant approximation
of the Lyapunov exponents and the Lyapunov dimension
(see also effects in \eqref{BenettinWrong}
for the approaches of Benettin~et~al.~\cite{BenettinGGS-1980-Part2} and Wolf~et~al.~\cite{WolfSSV-1985}).



\section{Conclusion} \label{sec:conclusion}
In this work the Rabinovich system,
describing the process of interaction between waves in plasma, is considered.
We show that the methods of numerical continuation and perpetual point
are helpful in localization and understanding of
hidden attractor in the Rabinovich system.
For the study of dimension of the hidden attractor
the notion of the \emph{finite-time Lyapunov dimension} is developed.
An approach to reliable numerical estimation of
the finite-time Lyapunov exponents (see relations \eqref{LEapprox}-\eqref{LEapproxuni})
and finite-time Lyapunov dimension (see relations \eqref{dimLmunest}) is suggested.
Various numerical estimates of the finite-time Lyapunov dimension
for the hidden attractor in the case of multistability are given.

\section*{Acknowledgements} \label{sec:acknowledgement}
The work in sec.~1-4 is done within the joint grant from DST and RFBR
(INT/RUS/RFBR/P-230 and 16-51-45002); 
in sec. 5-7 within  Russian Science Foundation project (14-21-00041).

\bibliographystyle{elsarticle-num}

\end{document}